\documentclass[a4paper,reqno]{amsart}

\usepackage[centertags]{amsmath}
\usepackage{amsfonts}
\usepackage{euscript}
\usepackage{amssymb}
\usepackage{amsthm}
\usepackage{newlfont}
\usepackage{stmaryrd}
\usepackage{mathrsfs}
\usepackage{euscript}
\usepackage{graphicx}
\usepackage{subfigure}
\usepackage{caption}
\usepackage[usenames,dvipsnames]{color}
\usepackage{enumerate}
\usepackage[foot]{amsaddr}

\theoremstyle{plain}
  \newtheorem{theorem}{Theorem}[section]
  
  \newtheorem{proposition}[theorem]{Proposition}
  \newtheorem{lemma}[theorem]{Lemma}
\theoremstyle{definition}
  \newtheorem{definition}{Definition}[section]
  \newtheorem{assumption}[]{Assumption}

  \newtheorem{fact}[]{Fact}

\theoremstyle{remark}
  \newtheorem{remark}[theorem]{Remark}

\numberwithin{equation}{section}

\let\al=\alpha \let\be=\beta \let\de=\delta 
  \let\ga=\gamma 
\let\ka=\kappa  \let\om=\omega 
\let\si=\sigma

 \let\Ga=\Gamma \let\La=\Lambda \let\Om=\Omega

\newcommand{\caA}{{\mathcal A}}
\newcommand{\caB}{{\mathcal B}}
\newcommand{\caC}{{\mathcal C}}

\newcommand{\caF}{{\mathcal F}}

\newcommand{\caH}{{\mathcal H}}

\newcommand{\caL}{{\mathcal L}}
\newcommand{\caM}{{\mathcal M}}

\newcommand{\caO}{{\mathcal O}}
\newcommand{\caP}{{\mathcal P}}
\newcommand{\caQ}{{\mathcal Q}}
\newcommand{\caR}{{\mathcal R}}
\newcommand{\caS}{{\mathcal S}}
\newcommand{\caT}{{\mathcal T}}

\newcommand{\caZ}{{\mathcal Z}}

\newcommand{\bbC}{{\mathbb C}}

\newcommand{\bbP}{{\mathbb P}}

\newcommand{\bbR}{{\mathbb R}}

\newcommand{\bbZ}{{\mathbb Z}}

\newcommand{\opunit}{\text{1}\kern-0.22em\text{l}}
\newcommand{\funit}{\mathbf{1}}

\newcommand{\frc}{{\mathfrak c}}

\newcommand{\frp}{{\mathfrak p}}

\newcommand{\frs}{{\mathfrak s}}

\newcommand{\frx}{{\mathfrak x}}

\newcommand{\frC}{{\mathfrak C}}

\newcommand{\frH}{{\mathfrak H}}

\newcommand{\frK}{{\mathfrak K}}

\newcommand{\frP}{{\mathfrak P}}

\newcommand{\frR}{{\mathfrak R}}
\newcommand{\frS}{{\mathfrak S}}

\newcommand{\frX}{{\mathfrak X}}

\newcommand{\frZ}{{\mathfrak Z}}

\newcommand{\bst}{{\boldsymbol t}}

\newcommand{\bsB}{{\boldsymbol B}}

\newcommand{\bsS}{{\boldsymbol S}}

\DeclareMathAlphabet{\mathpzc}{OT1}{pzc}{m}{it}

\newcommand{\ie}{i.e.\;}

\newcommand{\rel}{\,|\,}

\DeclareMathOperator{\Tr}{Tr}
\DeclareMathOperator{\tr}{tr}
\DeclareMathOperator{\spc}{sp}

\DeclareMathOperator{\skl}{skl}
\DeclareMathOperator{\rank}{rank}

\newcounter{smallroman}

 \newcounter{smallarabics}
\newenvironment{arabicenumerate}
{\begin{list}{{\normalfont\textrm{\arabic{smallarabics})}}}
  {\usecounter{smallarabics}\setlength{\itemindent}{0cm}
  \setlength{\leftmargin}{5ex}\setlength{\labelwidth}{4ex}
  \setlength{\topsep}{0.75\parsep}\setlength{\partopsep}{0ex}
   \setlength{\itemsep}{0ex}}}
{\end{list}}

\newcommand{\ben}{\begin{arabicenumerate}}
\newcommand{\een}{\end{arabicenumerate}}

\newcommand{\beq}{\begin{equation}}
\newcommand{\eeq}{\end{equation}}
\newcommand{\baq}{\begin{eqnarray}}
\newcommand{\eaq}{\end{eqnarray}}

\renewcommand{\sp}{{\mathrm{sp}}}

\renewcommand{\d}{\mathrm{d}}
\newcommand{\e}{\mathrm{e}}

\newcommand{\str}{ |}

\begin{document}

\title{Locality and nonlocality of classical restrictions of quantum spin systems\\ with applications to quantum large deviations and entanglement}

\author{W.~De Roeck${}^{1}$}
\address{${}^{1}$Instituut voor Theoretische Fysica, KU Leuven, Belgium.}
\email{wojciech.deroeck@fys.kuleuven.be}

\author{C.~Maes${}^{1}$}
\email{christian.maes@fys.kuleuven.be}

\author{K.~Neto\v{c}n\'y${}^{2}$}
\address{${}^{2}$ Institute of Physics AS CR, Prague, Czech Republic.}
\email{netocny@fzu.cz}

\author{M.~Sch\"utz${}^{1}$}
\email{marius.schutz@fys.kuleuven.be}

\begin{abstract}
We study the projection on classical spins starting from quantum equilibria. We show Gibbsianness or quasi-locality of
the resulting classical spin system  for a class of gapped quantum systems at low
temperatures including quantum ground states.   A consequence of Gibbsianness is the validity of a large deviation principle in the
quantum system which is known and here recovered in regimes of high
temperature or for thermal states in one dimension.   On the other hand we give an example of a quantum ground state with strong nonlocality in the classical restriction, giving rise to what we call measurement induced entanglement, and still satisfying a large deviation principle.
 \end{abstract}
\maketitle
\markboth{De Roeck, Maes, Neto\v{c}n\'y, Sch\"utz}{(Non-)Locality of classical restrictions} 
 
 
\section{Introduction}

The present paper investigates aspects of locality and nonlocality for states $\om$ of quantum spin systems, defined as thermal states or ground states of local Hamiltonians.  For that purpose we select a single site observable $X$ and consider its copy $X_i$ at each site $i$ of the $d$-dimensional lattice $\bbZ^d$.  The spectrum of $X$ is a finite set of eigenvalues $x \in \sp(X)$ and the  state $\omega$ naturally induces a probability distribution $\mu^X$  on $\sp(X)^{\bbZ^d}$.  Informally, for all finite sets $\Lambda\subset \bbZ^d$, the probability to find the values $x_i$, $i\in\La$, equals
\[
\mu^X[x_i, i\in \Lambda] = \omega\bigl(\prod_{i\in \Lambda}Q_i(x_i)\bigr), \quad x_i\in \sp(X)
\] 
where $Q_i(x_i)$ is a copy of the projection $Q(x)$ appearing in the spectral decomposition $X=\sum_{x} x\,Q(x)$.
Our main question is whether $\mu^X$ allows for a quasi-local description, for example in terms of a well-behaved potential such as for classical Gibbs distributions, more details are given below.  Obviously, the answer not only depends on the quantum state $\omega$ (and on all the parameters in its Hamiltonian) but possibly also on the chosen observable $X$.  Our results cover three cases:\\

\begin{enumerate}
\item  For high temperature quantum spin systems, the distribution $\mu^X$ is always Gibbsian. That is stated in Theorem \ref{thm: gibbs}.
\item  For low temperature and in the case of a unique ground state, we give in Theorem \ref{thm: gibbs1} sufficient conditions for the existence of an (exponentially decaying) potential making $\mu^X$ a Gibbs distribution, but
\item We also give counter examples (where the conditions are not satisfied), showing absence of quasi-locality  in $\mu^X$ for some $X$ and ground state $\omega$.
\end{enumerate}

These statements and precise results are introduced and discussed in the following three sections.  For the sake of concreteness we already illustrate them in the next section in the case of the quantum Ising model in a transverse field.\\

The motivation for the above questions is diverse and we come back to this point in the discussion of Section \ref{disc}. There are in fact two major applications.  The first is to the theory of large deviations for quantum spin systems.  The results of Theorems \ref{thm: gibbs} and \ref{thm: gibbs1}  imply the existence of a large deviation principle for sums of single site observables.  This high temperature result was already derived in \cite{NR}, relying in essence on similar expansion techniques as here.  The validity of low temperature and ground state large deviations is mostly new; we say more in Section \ref{disc}.\\
Secondly, and alternatively, the breaking of quasi-locality in $\mu^X$ implies a type of entanglement for the quantum ground states.  We call it $X$-measurement induced entanglement and it is related to `long range localizable entanglement', as introduced in \cite{PVC} to study questions similar to ours.  \\

For the plan of the paper, the next section discusses the results for the quantum Ising model.  The general framework gets introduced in Section \ref{genf} where the notion of classical restriction is most important.
Section \ref{resu} contains the main results, theorems and counter examples giving the more general version of what happens already in the quantum Ising model.  We also highlight there the dependence on the observable $X$ in case of low temperature and ground states.  Section \ref{disc} is devoted to discussion and more general background of motivations.  The proofs are collected in Sections \ref{prs1} and \ref{prs2}, and are written in a self-contained way.  An Appendix  recalls some facts in the analysis of the quantum Ising chain.

\subsection*{Acknowledgments} We thank Bruno Nachtergaele for useful discussions at the start of this project and Aernout Van Enter for careful reading of the manuscript and for correcting several errors. W.D.R. and M.S. are thankful to the DFG (German Research Fund) for financial support. C.M. gratefully acknowledges financial support in the form of an InterUniversity Attraction Pole  DYGEST (Belspo, Phase VII/18).

\section{Example: the quantum Ising chain}\label{sec: ExampleIsing}

The quantum Ising chain in a transverse magnetic field has formal Hamiltonian
\beq \label{quantum ising}
H = -J\sum_{i} \sigma_i^x\,\sigma_{i+1}^x - h\sum_i \sigma_i^z
\eeq
in one dimension ($i\in \bbZ$), and with the Pauli matrices $(\sigma_i^x,\sigma_i^y,\sigma_i^z)$ in the three directions as usual for spin 1/2 particles. The coupling $J$, the magnetic field $h>0$ and the inverse temperature $\beta$ parametrize the equilibrium state $ \Tr[\,\cdot\, \e^{-\beta H}/Z]$. In the limit $\beta\rightarrow \infty$ the model undergoes a quantum phase transition with critical point at $|J/h|=1$, see e.g.~\cite{sachdevbook}.
For $|J/h|\ll1$, the ground state is a perturbation of the  state 
 \beq  \label{symmetric groundstate} 
 \str \uparrow \rangle \otimes \str \uparrow \rangle \otimes \ldots \otimes \str \uparrow \rangle
   \eeq 
where $\str\uparrow \rangle$ is the normalized eigenvector of $\sigma_z$ with eigenvalue $+1$ and $\str \downarrow \rangle$ stands for the normalized eigenvector with eigenvalue $-1$.  Note that the state \eqref{symmetric groundstate}  is completely \emph{disordered} in the $\sigma^x$-basis: 
 \beq  \label{eq: freeness of sigmax}
 \str \langle\, \mathrm{a}\,  \str \uparrow \rangle \str^2 =\str \langle \,  \mathrm{b}\,  \str \uparrow \rangle \str^2=1/2 \eeq 
 where $\mathrm{a},\mathrm{b}$ stand for the two normalized eigenvectors of $\sigma_x$.

There are three natural choices for classical restrictions.  We can look at the probability distributions $\mu^x, \mu^y$ and $\mu^z$ obtained from the quantum equilibrium state by choosing  $X= \sigma^x, \sigma^y,\sigma^z$, respectively.

\begin{enumerate}
\item  The first type of results is in the regime  
$|J|,|h|\ll\beta^{-1}$ (high temperature); then all three spin-distributions $\mu^{x,y,z}$ are Gibbsian.
\item The second class of results is at low temperatures but needs extra conditions.  We
think of the transverse
magnetic field (second term in \eqref{quantum ising}) as the classical model with a small quantum
perturbation (first term). For that case our results show Gibbsianness for $\mu^x$ and $\mu^y$ , whenever $\beta^{-1},|J|\ll|h|$, including the ground state. This also implies a large deviation property for the macroscopic magnetizations $M^x_N = \sum_{i=1}^N \sigma_i^x/N$ and  $M^y_N = \sum_{i=1}^N \sigma_i^y/N$.
\item However, in the disordered ground state,  for $|J| \ll |h|$, the distribution $\mu^z$ is no longer local (in the sense that its local conditional distributions do not allow a continuous version) and hence not Gibbsian; see Theorem \ref{thm: nongibbs}.  Yet, a large deviation principle still holds for the magnetization  $M^z_N = \sum_{i=1}^N \sigma_i^z/N$; see Theorem \ref{thm: LdpForIsing}. 
\end{enumerate}


\section{Set-up}\label{genf}
A quantum spin system on the lattice $\bbZ^d$ is made from first associating to each site $i \in \bbZ^d$ a finite-dimensional Hilbert space $\caH_i$ as a copy of $\bbC^{m}$, $m=2,3,\ldots$ and the algebra of operators
$\caB(\caH_i)$, i.e.,  the $m\times m$ complex matrices.
In this section, $A,\La\subset \bbZ^d$ denote finite subsets of $\bbZ^d$, and we more generally write $A,\La\Subset \bbZ^d$ to indicate finiteness of subsets.
 The local Hilbert space for a volume $\La$ is the tensor product $\caH_\La = \bigotimes_{i \in
\La} \caH_i$, and $\cal A_\La= \caB(\caH_\La)=  \bigotimes_{i \in
\La}  \caB(\caH_i)$ denotes the local
matrix algebra. We employ the standard embedding $\caA_{\La'}\subset \caA_{\La}$, $\La'\subset\La$, through $M_{\La'}\otimes \opunit_{\La\setminus\La'}$ for $M_{\La'}\in \caA_{\La'}$.
The completion of $ \bigcup_{\La\Subset\bbZ^d} \caB(\caH_\La)$ in the operator norm defines
the (infinite volume) quasi-local algebra $\cal A$. As usual, a state is a normalized positive functional on this ($C^*$-)algebra $\caA$.

\subsection{Quantum equilibrium states}\label{qGs}
A (quantum) interaction, also sometimes called potential,
is a collection $\Phi = \{\Phi(A)\}$ of self-adjoint elements
$\Phi(A) \in \mathcal{B}(\mathcal{H}_A )$ labeled by
$A\Subset \bbZ^d$, where $\Phi(\emptyset) = 0$.
Throughout the article we assume translation invariance,
i.e., for all $A$ and $i \in Z^d$, $\Phi(A + i)$ is a copy of
$\Phi(A)$ acting on $\mathcal{H}_{A+i}$,
and also that $\Phi(A)=0$ whenever $A$ is not a connected set.
To prevent confusion, we mention that this assumption is
not made for the classical potential introduced further below. Each $\Phi(A)$ can also be regarded as local operator in $\caB(\caH_\La)$ for $A\subset\La\Subset\bbZ^d$. We use the norm
\begin{equation}\label{eq: norm}
  \|\Phi\|_\ka: = \sum_{A \ni 0} \e^{\ka |A|}\, \|\Phi(A)\|,\quad \ka \geq 0
\end{equation}
with $|A|$ counting the number of sites in $A$.  In particular, these norms are finite if the potential has finite range $r$, i.e.\  if $\Phi(A)$ vanishes whenever $A$ contains two sites at a (lattice)distance larger than $r$.

The local Hamiltonian is
\begin{equation}
  H_\La^\Phi = \sum_{A \subset \La} \Phi(A)
\end{equation}
and it defines the finite volume Gibbs state $\om_\La^\be$ at
inverse temperature $\be$ by
\begin{equation}\label{eq: local Gibbs}
  \om_\La^\be(\,\cdot\,) = \frac{1}{Z_\La^\be} \Tr_\La(\e^{-\be H_\La^\Phi}\,\cdot\,)\,, \qquad
  Z_\La^\be = \Tr_\La(\e^{-\be H_\La})
\end{equation}
with $\Tr_\La$ the standard trace on $\caB(\caH_\La)$.\\

The thermodynamic limit  $\La \nearrow \bbZ^d$ is taken along any sequence of volumes such that eventually $\Delta\subset \Lambda$ for any  $\Delta\Subset\bbZ^d$. Under suitable assumptions the 
states $\om_\La^\be$ have a  weak$^*$ limit $\om$ satisfying the Kubo-Martin-Schwinger {\small (KMS)} conditions in the standard sense of the quantum
equilibrium formalism; see~\cite{BR} for definitions and more
details. Furthermore, we can also define ground states in finite volume and take their thermodynamic limit. In all cases discussed in this paper, the ground state is unique and we can also obtain it by taking the (weak$^*$) $\be \to \infty$ limit of the infinite-volume states $\om=\om^{\be}$. Hence, the order of limits does not matter here.

\subsection{Classical restriction}\label{cr}
We choose a self-adjoint matrix $X \in
\caB(\caH)$ and write $X_i$ for its copies in $\caB(\caH_i)$, $i
\in \bbZ^d$. We also write $\Om_\La = \spc(X)^\La$ for the set of (classical) configurations in finite volume. Obviously, the collection $X_\La = (X_i)_{i\in \La}$ is a family of mutually commuting observables and we can define joint spectral projections $Q(x_\La)$, such that 
\begin{equation}\label{eq: spectral}
\prod_{i \in \La}F_i(X_i) = \sum_{x_\La \in \Om_\La}  (\prod_{i \in \La}F_i(x_i)) Q(x_\La)
\end{equation}
for all families of functions $F_i$ on $\spc(X)$. 

We now define the classical restriction of a state $\om$ as the probability distribution $\mu^X$ on $\Om_\La$ with probabilities
\begin{equation}\label{eq: spectral1}
  \mu^X(x_\La) = \om\bigl(Q_\La(x_\La)\bigr)
\end{equation}
 According to the quantum formalism, \eqref{eq: spectral1} gives the frequencies
 of outcomes when repeatedly and independently measuring the observables  $X_{i \in \La}$.
We do not indicate the dependence on $\La$ in  $  \mu^X$ since the family of probability distributions thus constructed is consistent and it defines a unique probability distribution on the infinite product $\Om:=\spc(X)^{\bbZ^d}$ (for the sake of precision: with Borel sigma algebra generated by the product topology on $\Om$). In other words, the probability distribution $\mu^X$ is a state for a classical spin system. 
That classical restriction $\mu^X$ depends of course on  the inverse temperature $\beta$, and on all other parameters in the quantum Hamiltonian, and sometimes we write $\mu^{\be,X}$ to emphasize this.  

Given a configuration $x \in \Om$ and for a volume $\La \subset \bbZ^d$, we write $x_{\La}$ for its restriction to $\Om_{\La}$. For finite $\La$ and  (not necessarily finite) $\La_1 \subset \La^\complement$,  we denote conditional probabilities by $\mu(x_{\La} \str x_{\La_1})$. 
By standard probability theory, these conditional probabilities are well-defined   for $\mu$-almost every $x_{\La_1} \in \Om_{\La_1}$.
\begin{remark}\label{rem: nullity} Classical restrictions for quantum ground states can easily show a property  called `nullness'.  As an example take the ground state of the transverse Ising model at $J=0$ and $h>0$, i.e.\ 
\eqref{symmetric groundstate}. We  choose the observable $X$ to be
\beq\label{rema}
  X =  \str \uparrow \rangle \langle \uparrow \str -  \str \downarrow \rangle \langle \downarrow \str
\eeq
having eigenvalues $\pm 1$.  It is obvious that the classical restriction $\mu^X$ satisfies $\mu^X(x_i=-1)=0$ for all sites $i$. 
\end{remark}

 \subsection{Gibbsianness and quasi-locality}\label{subsec: Gibbsianess}
We consider now probability distributions on the configuration space $\Om$. A family $\Psi =
\{\Psi_A\}$, $A\Subset \bbZ^d$, of functions $\Psi_A: \Om_A
\rightarrow \bbR$ with $\Psi_{\emptyset}=0$ is called a (classical) potential. Here we always consider potentials that are translation invariant and we make use of the following norms, cf.~\eqref{eq: norm} for the quantum analogue,
\begin{equation}\label{eq: norm-cl}
  \|\Psi\|_\ka :=  \sum_{A \ni 0} \e^{\ka |A|} \sup_{x_\La \in \Om_\La} |\Psi_A(x_\La)|,\qquad \ka\geq 0
\end{equation}

\begin{definition}\label{def: gibbs}
A probability distribution $\mu$ on $\Omega$ is Gibbsian  if there is a classical
potential $\Psi$ with  $\|\Psi\|_0 < +\infty$  such that for every
$\La$ and for $\mu$-almost every
$x_{\Lambda^\complement} \in \Om_{\Lambda^\complement}$,
\begin{align}
\mu(x_\La \rel x_{\Lambda^\complement})
  &= \frac{1}{\caZ_\La(x_{\Lambda^\complement})} \exp
  \bigl[ -\sum_{A \cap \La \neq \emptyset} \Psi_A(x)
  \bigr]\label{ql}
\\ \intertext{with}
  \caZ_{\La}(x_{\Lambda^\complement}) &= \sum_{x_\La \in \Om_\La}
  \exp \bigl[ - \sum_{A \cap \La \neq \emptyset} \Psi_A(x)
  \bigr]
\end{align}
\end{definition}
Note that we avoided hard core interactions (with $\Psi_A$ that can take the value infinity at some configurations).
For a general theory of Gibbs distributions we refer to~\cite{EFS, F,georgii}.

From Definition \ref{def: gibbs} one sees that a Gibbs distribution $\mu$ is \emph{quasi-local} in the sense that it allows a version for its local conditional distributions that is continuous; see \eqref{ql}
where the right-hand side only weakly depends on far away spins. In fact, a probability distribution $\mu$ on $\Om$ is Gibbsian if and only if its system of conditional probabilities $\mu(x_\La \rel x_{\Lambda^\complement})$ has a version that is both continuous (`quasi-locality') and positive (often called `non-null' in this context). This result goes back to \cite{Koz, Su}. Probability distributions that are not quasi-local have configurations of essential discontinuity:\\
A configuration $x\in \Omega$ is \underline{bad} for a probability distribution $\mu$ if there is $\varepsilon >0$ and $i\in \bbZ^d$ so that for all $\Lambda\Subset \bbZ^d$, $i\in\La$, there is a finite volume $\Gamma \supset \Lambda$ and there are configurations $y,y'\in \Omega$ with $\mu(x^{}_{\Lambda\setminus i}\,y^{}_{\Gamma\setminus \Lambda}),  \mu(x^{}_{\Lambda\setminus i}\,y'_{\Gamma\setminus \Lambda}) >0$ such that
\beq \label{eq: locality concrete}
 \left\str \mu(x^{}_i \,\str \, x^{}_{\Lambda\setminus i}\,y^{}_{\Gamma \setminus \Lambda} ) - 
 \mu(x^{}_i \,\str \, x^{}_{\Lambda\setminus i}\,y'_{\Gamma \setminus \Lambda} )\right\str > \varepsilon
\eeq
In words, the state at site $i$ conditioned on the values of spins in $\Lambda \setminus i$ keeps depending on additional conditioning outside $\Lambda$ no matter how big that volume $\Lambda$ is.

Finally, it is important that without further conditions the translation-invariant Gibbs distributions of Definition \ref{def: gibbs} satisfy a large deviation principle, see e.g.\cite{lanford,georgii}, implementing the static fluctuation theory that forms the basis of equilibrium statistical mechanics. 

\section{Results}\label{resu}

For any suitably decaying quantum interaction $\Phi$ there is a unique equilibrium state $\om^\be$ satisfying the {\small KMS} conditions for high enough temperatures $1/\be$, see e.g.\ \cite{BR}. This state $\om^\be$ is the thermodynamic limit of finite volume Gibbs states $\om_\La^\be$, see \eqref{eq: local Gibbs}, and in particular its classical restriction $\mu^{\be,X}$ can be obtained as
\begin{equation}\label{eq: DefFinClaRes}
  \mu^{\beta,X}=\lim_{\Lambda\nearrow\bbZ^d}\mu^{\beta,X}_\Lambda\quad\text{with}\quad \mu^{\beta,X}_\Lambda(x_\Lambda):=\omega_\Lambda^\beta\bigl( Q_\Lambda(x_\Lambda) \bigr),\; x_\Lambda \in \Omega_\Lambda
 \end{equation}
\begin{theorem}[High temperature]\label{thm: gibbs}
 Let $\Phi$ be an interaction with $\lVert \Phi\rVert_\ka<\infty$ for a given $\ka>0$. Then there exists $\be_\mathrm{max}>0$ such that the classical restriction $\mu^{\be,X}$ of the (unique) quantum equilibrium state $\om^\be$ is Gibbsian for $\be \leq \be_\mathrm{max}$ and for every self-adjoint matrix $X \in \caB(\caH)$.
\end{theorem}
In the proof, see \eqref{eq: high-T condition}, we give an explicit estimate of an inverse temperature $\beta_0>0$, such that the thermodynamic limit \eqref{eq: DefFinClaRes} exists for all $\be\leq\be_0$. \\
There are various properties of the resulting large deviation rate function that follow.  As of independent interest, at
high temperatures, these results can be used to obtain a central
limit theorem; we refer to ~\cite{NR} for further discussion.\\

At low temperatures we specify the regime in which our results hold by two \textbf{assumptions}; the first is concerned with the interaction underlying the quantum state, the second spells out a condition on the single-site observable $X$ which induces the classical restriction.

\begin{assumption}
Suppose an interaction $\Phi = \Phi_0 + \Upsilon$, where $\Phi_0$ has finite range. Assume there is a one-dimensional orthogonal projection $\caP\in\caB(\caH)$, such that the local Hamiltonian $H_\Lambda^{\Phi_0}$ satisfies the following, for all $\Lambda$ and  $S\subset \Lambda$:
\begin{enumerate}
 \item $H_\Lambda^{\Phi_0}$ commutes with $\caP_{\La}(S)$,
 \item $H_\Lambda^{\Phi_0} \caP_\Lambda(\emptyset)=0$,
 \item there is a $\Lambda$-uniform gap $g>0$, such that 
 \begin{equation}\label{eq: DecayOfExcitations} H_\Lambda^{\Phi_0}\caP_\Lambda(S)\geq g\,|S|\caP_\Lambda(S)\end{equation}
 in the sense of positive operators,
\end{enumerate}
where we defined the projections
\begin{equation}
 \caP_\Lambda(S):=\bigl({\textstyle \bigotimes_{i\in S}}\caP_i^\perp\bigr)\otimes \bigl({\textstyle\bigotimes_{i\in \Lambda \setminus S}}\caP_i^{}\bigr)
\end{equation}
in $\caB(\caH_\Lambda)$.
\end{assumption}

The condition \eqref{eq: DecayOfExcitations} is a Peierls condition: the local Hamiltonians $H_\Lambda^{\Phi_0}$ have a ($\Lambda$-uniformly) gapped non-degenerate product ground state.  As an example, we look  at the disordered ground state \eqref{symmetric groundstate} in the quantum Ising model
of Section \ref{sec: ExampleIsing}.  We can take there $\cal P=\lvert\uparrow \rangle \langle \uparrow \rvert$, and $\Phi_0$ corresponds to the second term in the Hamiltonian \eqref{quantum ising} (transverse field).  In our treatment the second term $\Upsilon$ will be a sufficiently small perturbation of the particularly simple interaction $\Phi_0$. In this case the above assumption implies a unique ground state for the interaction $\Phi$, see e.g.~\cite{Ya1}, and furthermore applicability of so-called quantum Pirogov--Sinai theory, see \cite{BKU,DFF}. As a consequence there is a unique {\small KMS} state for small enough temperatures and the classical restriction $\mu^{\beta,X}$ can again be obtained through \eqref{eq: DefFinClaRes}. \\
  
There is a second major assumption: the first term $\Phi_0$ must not in any way `fix'
the observable $X$; it must remain `free' and sufficiently unbiased in the presence of that dominant term:

\begin{assumption}
Suppose that $\Tr\bigl(Q(x) \caP\bigr) >0 $ for all $x\in \sp(X)$. 
\end{assumption}
Clearly, that is not satisfied in the case of the Ising model for $X=\sigma^z$ and $\cal P=\lvert\uparrow \rangle \langle \uparrow \rvert$ as above.  There is however then no problem in the case of $X=\sigma^x$ or $X=\sigma^y$; they are left `free'; see in particular \eqref{eq: freeness of sigmax}.

\begin{theorem}[Low temperature \& weak coupling]\label{thm: gibbs1}
Take the {\bf Assumptions 1 and 2} above.  There exist positive $\kappa_{\min},\beta_{\min}$ (depending on $X$) so that if $\kappa\geq\kappa_{\min}$, $\beta\geq\beta_{\min}$ and  $\lVert \Upsilon \rVert_\kappa\leq 1$, then $\mu^{\beta,X}$ is Gibbsian. Moreover, this statement remains true for the ground states, i.e.\ for  $\beta \rightarrow\infty$.
\end{theorem}

The most striking condition in the above theorem is {\bf Assumption 2}, which in particular excludes observables $X$ that commute with the projector $\caP$. A first reason for it is to avoid the nullness-scenario mentioned in Remark \ref{rem: nullity}, which rules out Gibbsianness right away. Note that there the quantum ground state (and classical restriction) is local as a product state. At least at zero temperature ($\be=\infty$), {\bf Assumption 2} can surely not be dropped also in view of the more interesting quasi-locality aspect of Gibbsianness, as follows from the following.

\begin{theorem}[Non-quasi-local ground state]\label{thm: nongibbs}
Consider the Ising model in transverse field as discussed in Section \ref{sec: ExampleIsing} and let  $X=\si^z$. Let  $\be=\infty$ and $|J/h|>0$ be small enough.  Then, the corresponding classical restriction $\mu^z$ is nonnull in the sense that $\mu^z(x_\La)>0$ for any $x_\La \in \Om_{\La}$, $\La\Subset\bbZ$.  Most importantly, $\mu^z$ is not quasi-local and the configuration $x\in\Om$ defined by $x_i=-1$, $i\in \bbZ$, is a bad configuration.
 \end{theorem}
As pointed out to us by Aernout van Enter,
the computations in the proof of the Theorem can be used to show
that in fact all configurations $x\in \Omega$ are bad for $\mu^z$.
For simplicity we supply the explicit proof only for the configuration
$x\equiv -1$ as in the Theorem. The result holds for higher dimensions $d>1$ as well,
as one checks by going through the proof,
but again we restrict ourselves to d = 1 for brevity.
The fact that a classical restriction of the ground state is not quasi-local does not mean that it does not satisfy a large deviation principle, as we see in 

\begin{theorem}[Large deviation principle despite Non-Gibbsianness]\label{thm: LdpForIsing}
As in Theorem \ref{thm: nongibbs}, consider the transverse Ising model in the disordered regime $|J/h|<1$ with $X=\si^z$, $\be=\infty$. Then the generating function
 \begin{equation}
  F(t):=\lim_{n\rightarrow \infty}{\textstyle \frac{1}{n}}\log \omega\left( \exp\bigl( t{\textstyle \sum_{i=1}^n} \sigma_i^z \bigr) \right), \qquad t \in \bbR
 \end{equation}
 exists and is real-analytic. 
\end{theorem}
As a consequence of the G\"artner-Ellis theorem, see e.g.~\cite{E}, Theorem \ref{thm: LdpForIsing} implies that in the disordered  ground state of the quantum transverse Ising model the magnetization $M_N^z:= \frac{1}{N}\sum_{i=1}^N \sigma_j^z$ satisfies a large deviation principle.  More precisely, with respect to the classical restriction $\mu^{z}$, $m_N = \frac 1{N} \sum_{j=1}^N x_j$ as a function on $\Om$ satisfies a large deviation principle for a (lower semi-continuous and convex) rate function $I$ which is the Legendre transform of $F$:
 \begin{equation}
  \begin{split}
   \limsup_{n\rightarrow \infty} \frac{1}{n} \log \mu^{z} \bigl(m_n \in C\bigr)&\leq -\inf_{m\in C} I(m) \quad\text{ for } C\subset \bbR \text{ closed}\\
    \liminf_{n\rightarrow \infty} \frac{1}{n} \log \mu^{z} \bigl(m_n \in O\bigr)&\leq -\inf_{m\in O} I(m) \quad\text{ for } O\subset \bbR \text{ open}
  \end{split}
 \end{equation}

\section{Discussion}\label{disc}

The issue of (non-)locality of classical restrictions of quantum states $\omega$ has two major applications for $\omega$.  The (quasi-)locality (such as per consequence of Theorems \ref{thm: gibbs} and \ref{thm: gibbs1}) implies well behaved large deviations and the non-locality of classical restrictions of quantum ground states (such as per consequence of Theorem \ref{thm: nongibbs}) implies some strong form of entanglement in that quantum ground state.

\subsection{Fluctuation theory}\label{qft}

Fluctuation theory, or the theory of large deviations \cite{DZ,E}, remains important  when moving to the quantum regime, e.g. for a relevant understanding of variational principles and of response theory, see e.g.~\cite{DMN}.   Let $F$ be a function on $\sp(X)$ and consider the spatial average
\begin{equation}\label{emp}
  \bar F_\La = \frac{1}{|\La|} \sum_{i \in \La} F(X_i)
\end{equation}
  Fluctuation theory is about characterizing the `probabilities'
$\om \left( 
  \chi_{[a,b]}(   \bar F_\La ) \right)$,  where
$\chi_{[a,b]}(\cdot)$ denotes the indicator function of some interval $[a,b] \subset \bbR$.  That gives the distribution of
the outcomes when measuring the average \eqref{emp}.  The point is that these fluctuations can be expressed via the classical restriction $\mu^X$, namely
\begin{equation}\label{cladev}
\om \left( 
  \chi_{[a,b]}(   \bar F_\La ) \right) = \mu^X \big( a\leq \frac 1{|\Lambda|}\sum_{i \in \La} F(x_i)  \leq b  \big)
\end{equation}
Hence the question emerges whether a large deviation principle holds for the distribution $\mu^X$.  But from classical statistical mechanics the answer is an immediate `yes' for equilibrium distributions.
Therefore Gibbsianness of the classical restriction $\mu^X$ of quantum equilibrium or ground states implies a (quantum) large deviation result. The results of the present paper, in particular Theorems \ref{thm: gibbs} and \ref{thm: gibbs1} thus add  to  the current state-of-the-art on quantum large deviations: they are now proven for:\\[1mm]
\noindent \emph{High temperature:} \normalfont see \cite{NR, LR}.\\[1mm]
\emph{Dimension $d=1$, be it quantum equilibrium states or finitely correlated states: }\normalfont   see \cite{O}. \\[1mm]
\emph{Low temperature or ground states with appropriate conditions:} \normalfont  the present paper, Theorem \ref{thm: gibbs1}. \\[1mm] 
In all these cases, the result is strong enough to imply a central limit theorem, because the large deviation generating function is analytic in a a neighborhood of $0$, but we give no further details.\\

A final remark concerns the property of \emph{asymptotic decoupling}, which is weaker than Gibbsianness, but stronger than large deviations, see  \cite{P} for definitions and proofs.  Therefore, in the present context, the asymptotic decoupling of $\mu^X$ suffices for quantum large deviations of \eqref{emp} in the quantum state. Such an asymptotic decoupling can indeed be shown at high temperature and in one dimension, see \cite{Oluc}. 

\subsection{$X$-Measurement-Induced Entanglement}
In this section we connect with notions of entanglement and it is therefore natural to restrict the discussion to pure states $\om$ even though the mathematics below allows generalizations to mixed quantum states.\\ 
We  `condition' the state $\om$ on the measurement outcome $x_V$ of the observables $X_i$, $i \in V\Subset\bbZ^d$, by defining: 
\beq \label{eq: conditioned state}
\om^{x_{V}} (\,\cdot\,) :=  \frac{\om( Q({x_{V}}) \cdot Q({x_{V}})  )}{\om(Q({x_{V}}))}
\eeq
Recall that we write $O_i$ for the local operator acting non-trivially on $\caH_i$ as copy of $O\in \caB(\caH)$. We say that the state $\om$ has \emph{`$X$-Measurement-Induced Entanglement'} whenever 
there are single-site observables  $A,B \in \caB(\caH)$ and a configuration $x \in \Om$ such that 
\begin{equation}\label{sent}
\limsup_{n \to \infty} | \om^{x_{V_n}}(A_0 B_{i_n}) - \om^{x_{V_n}} (A_0)\, \om^{x_{V_n}}( B_{i_n})|  >0 
 \end{equation}
for a sequence of punctured balls $V_n = \{i\,  \str \,   0< \str i-0 \str< n \}$ and  a sequence of sites $i_n \in V_n^\complement$. 
Hence, measuring $X$ in large regions $V$ can correlate observables that are spatially separated (It might be natural to allow that $A,B$ live on a few sites, rather than one, one can easily modify the definition in this direction).
 Physically we can imagine that in a region $V_n$ surrounding the center of a spin system a very strong magnetic field is applied to let the spins all point there in the same (field)direction; still the quantum ground state does not factorize for joint observations in the center and outside $V_n$. That notion is of course tailored towards the strong breaking of quasi-locality in the sense of \eqref{eq: locality concrete}. \\
 
\begin{fact} If there is a bad configuration $x \in \Om$  for the classical restriction $\mu^X$, then  $\om$ has `$X$-Measurement-Induced Entanglement'; see \eqref{eq: locality concrete}. 
\end{fact}

Indeed, by choosing $A,B$  in \eqref{sent} functions of $X$,  this is immediate from \eqref{eq: locality concrete}. The converse is not true: 
\begin{fact}  Quasi-locality of a classical restriction $\mu^X$ for some $X$ does not imply the  absence of `$X$-Measurement-Induced Entanglement'.
\end{fact}

The point is that in \eqref{sent} there remains extra freedom in the choice of $A$ and $B$, which do not need to be `classical' observables (commuting with $X$). We give an example below. 
However, let us first point out the difference with a related notion introduced in \cite{VMC}, namely `Long Range Localizable Entanglement' (LRLE): a state $\om$ has LRLE whenever the deviation from a product state in  \eqref{sent} is present for \emph{typical} configurations $x$. To implement this idea, one chooses some entanglement measure
of the conditioned state $\om^{x_{V_n}}$ and one averages that quantifier over $x_{V_n} \in \Om_{V_n}$ before taking $n \to \infty$. 
A somewhat surprising property, reinforcing {\bf Fact 2}, is that in case there is LRLE, there is a tendency for $\mu^X$ to be product, hence in particular local.  We do not build the framework to state this precisely but it is illustrated by our example below. 
It would also be interesting to investigate whether quantum states `typically' have or do not have `X-Measurement Induced Entanglement' for some observable $X$.  

\subsubsection{Example}  The  class of examples here includes the ground state of the AKLT model \cite{AKLT}. 
Let $\str \alpha\rangle, \al=1,\ldots,m$ be an orthonormal basis in the single-site Hilbert space $\caH\equiv\bbC^m$.
Let $A_{\alpha}$ be two $2\times 2$ matrices satisfying the following conditions
\begin{enumerate}
\item Up to multiplication with a complex number, $A_{\al}$  are unitaries. 
\item  $\sum_{\al} A^*_{\al}A_{\al} =\opunit $.
\item  The algebra generated by $A_{\al}, \al=1,\ldots,m$ is the full $2\times 2$ matrix algebra. 
\end{enumerate}
Then we define the following translation-invariant finitely correlated state \cite{FNW},
\beq \label{eq: finitely general}
\om(O_1O_2 \ldots O_{\ell}) =  \frac{1}{2}  \Tr_{\bbC^2}  \left[ E_{O_\ell}\circ    \ldots \circ E_{O_2}\circ E_{O_1}(\opunit)\right], \qquad  O_i \in \caB(\caH_i)
\eeq
where $E_{O_i}: \caB(\bbC^2) \to \caB(\bbC^2)$ is the map defined by 
\beq
E_{O_i}(D)=     V (D \otimes O) V^*, \qquad \text{with}\,\, V= \sum_{\al }A_{\al} \otimes \langle \al \str  \, \in \caB\bigl(\bbC^2 \otimes \caH_i\,,\bbC^2\bigr)
\eeq
Then, the constraints (2) and (3) above guarantee that the infinite-volume state $\om$ is a pure state with exponential decay of correlations; we refer to \cite{FNW2} for details.
We choose the observable $X= \sum_{\al} \al \str \al \rangle \langle \al \str$. 
One can now consider the conditioned state \eqref{eq: conditioned state} for a given $x$ and, using constraint (1), find that \eqref{sent} fails \emph{for every choice of $x$.}  Details of this calculation can be found in \cite{W}.    In  the language introduced above, this means that the state $\om$ has LRLE. Moreover, it has been shown that within a given class of finitely correlated states (namely those with `ancilla dimension' equal to $2$, i.e., corresponding to the fact that $A_\al$ are $2 \times 2$ matrices), this is the only example having LRLE. 
On the other hand, the classical restriction $\mu^X$ is a product measure.   To check this, it suffices to note that 
\beq
E_{\str \al \rangle \langle \al \str}(\opunit) \propto \opunit, \qquad    \text{for every $\al=1,\ldots,m$}
\eeq
and the product then property follows readily from \eqref{eq: finitely general}.

\section{High-temperature regime}\label{prs1}

In this section we give a proof of Theorem~\ref{thm: gibbs},
together with some more explicit formul{\ae} and estimates on the
classical potential $\Psi^{\be,X}$. The decisive step of our strategy, namely Proposition~\ref{prop: mu}, is based on a formulation of the problem in terms of a polymer model and on a perturbative construction by means of a high-temperature cluster expansion. We closely follow Section 6 of \cite{NR}.

\subsection{Logarithm of the classical restriction}

We start by deriving explicit formul{\ae} for the logarithm of $\mu_\Lambda^{\beta,X}$, the classical restriction of $\omega_\Lambda^\beta$, see \eqref{eq: DefFinClaRes}, for any $\Lambda\Subset \bbZ^d$. In this section we mostly suppress the dependence on the chosen single-site observable $X$ and on the inverse temperature $\be$.  \\
The symbol $\tr$ is used to denote the (normalized) trace state on $\caA$, and for $W\Subset \bbZ^d$ and for configurations $x_W \in \Om_W$ we write
\begin{equation}
  \tr^{x_W}(\,\cdot\,) = \frac{\Tr_W(\,\cdot\, Q_W(x_W))}{\Tr_W(Q_W(x_W))}
\end{equation}
which is a (normalized) state on $\caA_W$. By embedding it also defines a state on the quasi-local algebra $\caA$. From \eqref{eq: DefFinClaRes} we express the distribution $\mu_\Lambda$ 
in terms of these trace states:
\begin{equation}\label{eq: marginal}
\begin{split}
  \mu_\Lambda(x_W) &= \om_\Lambda(Q_W(x_W))
  = \frac{1}{Z_\Lambda} \Tr_\Lambda(e^{-\be H_\Lambda^\Phi} Q_W(x_W))
\\
  &= \frac{\Tr_\Lambda(Q_W(x_W))\,\tr^{x_W}(e^{-\be H_\Lambda^\Phi})}{\Tr_\Lambda(e^{-\be H_\Lambda^\Phi})}
\\
  &= f
  \frac{\tr^{x_W}(e^{-\be H_\Lambda^\Phi})}{\tr(e^{-\be H_\Lambda^\Phi})}
\end{split}
\end{equation}
The logarithm of the above finite volume partition functions can be written as a sum over local weights,
\begin{equation}\label{eq: PartFunc}
 \begin{split}
  \log\;\tr\bigl(e^{-\be H_\Lambda^\Phi}\bigr) &=  \sum_{A \subset \Lambda} w(A) \\
  \log\;\tr^{x_W}\bigl(e^{-\be H_\Lambda^\Phi}\bigr) &= \sum_{A \subset \Lambda} w^{x_W}(A)
 \end{split}
\end{equation}
where for all $\La\Subset \bbZ^d$ the weights are given as
\begin{equation}\label{eq: defweights}
 \begin{split}
  w^{}(A)&=\sum_{B\subset A} (-1)^{\lvert A\setminus B \rvert}  \log\;\tr\bigl(e^{-\be H_B^\Phi}\bigr)\\
  w^{x_W}(A)&= \sum_{B\subset A} (-1)^{\lvert A\setminus B \rvert}  \log\;\tr^{x_W}\bigl(e^{-\be H_B^\Phi}\bigr)
 \end{split}
\end{equation}
which goes by the name of `inclusion-exclusion principle', an application of more general M\"obius inversion theory. \\
Note that the weights are uniquely determined by the consistency requirement that the above equations hold for all $\La\Subset\bbZ^d$ for weights $w^{x_W}(A)$ which only depend on $A$ but not on the ambient volume $\La$. Furthermore $w^{x_W}(A)=w^{x_{W\cap A}}(A)$ and in particular we have $w^{x_W}(A)=w^{}(A)$, whenever $W\cap A=\emptyset$. We always write $w^{x_A}(A)$ instead of $w^{x_W}(A)$ if $A\subset W$.\\

\subsection{Gibbsianness -- proof of Theorem \ref{thm: gibbs}}
With the preceding definitions we can write $\mu^{}_\Lambda$ as Gibbs distributions for (finite-volume) classical potentials $\{\Psi_A^{}\}_{A\subset\Lambda}$, which are consistent for different $\La\Subset\bbZ^d$. One computes
\begin{equation}\label{eq: explicit mu}
\begin{split}
  \mu_\Lambda^{}(x_W) &= \mathrm{tr}\bigl( Q_W(x_W) \bigr)  \exp \Bigl( \sum_{\substack{A \subset \Lambda \\ A \cap W \neq \emptyset}}  \bigl[ w^{x_W}(A) - w^{}(A) \bigr] \Bigr)\\
   &=\sum_{x_{\Lambda\setminus W}\in\Omega_{\Lambda\setminus W}}\frac{1}{Z_{\Lambda}^{}} \exp \Bigl( \sum_{\substack{A \subset \Lambda}}  \Psi_A^{}(x_A) \Bigr)
\end{split}
\end{equation}
with
\begin{equation}
 \Psi^{}_A(x_A)=\left\{\begin{array}{ll}
                     w^{x_A}(A) + \log\bigl(\mathrm{tr}_i(Q_i(x_i)\bigr) & \textnormal{ for } A=\{i\}, i\in\bbZ^d \\[\baselineskip]
                     w^{x_A}(A)  & \textnormal{ else}
                    \end{array}\right.
\end{equation}
If the family $\Psi^{}=\{\Psi^{}_A\}_{A\Subset\bbZ^d}$ were a (infinite-volume) potential, i.e.~$\lVert \Psi \rVert_0<\infty$, then we could immediately conclude that any thermodynamic limit point of the $\mu_{\Lambda}^{}$ is a Gibbs distribution for $\Psi^{}$, see e.g.~\cite{Ru}. For high enough temperatures we know that there is indeed only the unique limit point $\mu^{}$ but an explicit demonstration of the convergence of \eqref{eq: explicit mu} is also contained in the proof of the next proposition.\\
Therefore, proving that $\Psi^{}$ and $w$ are potentials (for sufficiently high temperatures) will finish the proof of Theorem \ref{thm: gibbs}.
The relevant bounds on the weights are not deduced from their implicit definition given in \eqref{eq: defweights} but rather from a concrete construction obtained within the (non-commutative) Mayer-expansion formalism.

\begin{proposition}\label{prop: mu}
 For a quantum interaction $\Phi$, let $a, \be_0 > 0$ be such that
\begin{equation}\label{eq: high-T condition}
  \sum_{A \ni 0} e^{a |A|} (e^{\be_0 \|\Phi(A)\|} - 1) \leq a
\end{equation}
(which can always be achieved if $\lVert\Phi\lVert_\ka<\infty$ for a $\ka>0$) then
\begin{equation}
 \lVert w \rVert_0 \leq a\quad\text{and} \quad \lVert \Psi^{} \rVert_0 \leq a+\log(\dim{\caH})
\end{equation}
for all $\beta\leq\beta_0$. In particular $\Psi^{}$ is a potential for the Gibbs distribution $\mu^{}$ obtained as the unique thermodynamic limit of the $\mu_\Lambda^{}$. $\Psi_A$ is analytic in the open disk $\lvert\be\rvert<\beta_0$ for any $A\Subset\bbZ^d$.
\end{proposition}

\begin{proof}

The weights $w$ and $w^{x_W}$, which were used to express the finite volume partition functions with respect to the states $\mathrm{tr}$ and $\mathrm{tr}^{x_W}$, see \eqref{eq: PartFunc}, can be expressed more explicitly analoguously as it was done in \cite{NR} for states defined in eq.~(6.2) of this reference. There, the $\La$-uniform construction does not depend on the explicit form of the underlying state apart from the product property which is also shared by $\mathrm{tr}$ and $\mathrm{tr}^{x_W}$. In particular, the analogue of Proposition 6.25 in \cite{NR} remains valid. It is subject to condition (6.22) in \cite{NR} which is literally equal to our assumption \eqref{eq: high-T condition}. Transferred to our setting and accounting for the fact that the state $\mathrm{tr}^{x_W}$ is not translation invariant the result of Proposition 6.25 reads $\lVert w\rVert_0<a$ and 
\begin{equation}
 \sup_{i\in\bbZ^d}\sum_{ A\Subset\bbZ^d;\, i\in A}\bigl\lvert w^{x_W}(A)\bigr\rvert<a
\end{equation}
for all $x_W\in\Omega_W$, $W\Subset\bbZ^d$, and $\beta\leq\beta_0$. Therefore \eqref{eq: explicit mu} is convergent and the measures $\mu^{}_\La$ have a unique thermodynamic limit. Furthermore it follows that
\begin{equation}
 \sup_{i\in\bbZ^d}\sum_{ A\subset W; \,i\in A}\bigl\lvert w^{x_A}(A)\bigr\rvert<a
\end{equation}
for all $x\in\Omega$, $W\Subset\bbZ^d$, which shows that $\Psi$ is a potential.\\
Analyticity of each $\Psi_A$ follows from the arguments given in section 6.3.~of \cite{NR}.
\end{proof}

%
%
%

\allowdisplaybreaks

\section{Low temperature regime}\label{prs2}

Here we consider low temperatures and our perturbation strategy goes via an expansion around the ground state. Such expansions are familiar in the framework of quantum Pirogov-Sinai theory, \cite{BKU,DFF}; see also~\cite{Ya} for an alternative approach at zero temperature.

First, we show that the classical restrictions $\mu_\Lambda^{\beta,X}$, $\Lambda\Subset \bbZ^d$, are Gibbs distributions with a classical potential $\Psi_\Lambda^{\beta,X}$, which then by its explicit construction allows to control the thermodynamic limit as well as the limit of zero temperature. 
To lighten the notation we again often do not indicate the dependence on the volume $\Lambda\Subset\bbZ^d$, on the observable $X\in\caB(\caH)$, and on the inverse temperature $\beta>0$ whenever it does not inflict confusion. We also introduce the more convenient abbreviations $H_0\equiv H_\Lambda^{\Phi_0}$ for the local Hamiltonian, $V\equiv H_\Lambda^{\Upsilon}$ for its perturbation, and $H\equiv H_\Lambda^\Phi$ for the sum of both.\\
\subsection{Polymer model representation of the finite volume classical restriction}\label{sub: PolymerModel}

In finite volumes the finite-dimensional matrix exponential $\exp(-\beta H)$ can be written as its norm-convergent Dyson series:
\begin{equation}
\begin{split}
 & \e^{-\beta(H_0+V)}\\ 
 & = \e^{-\beta H_0} + \sum_{n\geq 1} \int_{\caS_n} \mathrm{d}t_1\dots\d t_n \, \e^{-(\beta-t_n) H_0} V \dots V \e^{-(t_2-t_1) H_0} V \e^{-t_1 H_0}\\
 &=\e^{-\beta H_0} + \sum_{n\geq 1}\sum_{\substack{S_0,\dots,S_n\\ \subset \Lambda}}\sum_{\substack{B_1,\dots,B_n\\ \subset \Lambda}} \int_{\caS_n} \mathrm{d}t_1\dots\d t_n \, \bigl(\caP_\Lambda(S_n)\e^{-(\beta-t_n) H_0} \Upsilon(B_n)\bigr) \dots \\
 &\hspace{.4 \textwidth}\dots \bigl(\caP_\Lambda(S_1)\e^{-(t_2-t_1) H_0}\Upsilon(B_1)\bigr)\e^{-t_1 H_0}\caP_\Lambda(S_0)\\[-.5\baselineskip]
 &
 \end{split}
\end{equation}
where $\bst_n=(t_1,\dots,t_n)$ and where we integrate over the simplex
\begin{equation}\caS_{n}:=\{\bst_n\in [0,\beta]^n\,|\,0< t_1 < t_2 <\dots <t_n< \be \}\end{equation} In the third line we have expanded the individual terms of the series by inserting the decomposition $\opunit=\sum_{S\subset \Lambda}\caP_\Lambda(S)$. 

We introduce new notation to reorganize the above expansion. Let $\bsS_n=(S_0,\dots,S_n)$ and $\bsB_n=(B_1,\dots,B_n)$, then, for $n\geq 0$, we define the set of \textit{diagrams of $n$ interactions} $\frS_n$. For $n\geq 1$,
\begin{equation}
\begin{split}
\frS_n:=\bigl\{\bigl(\bst_n,\bsS_n,\bsB_n\bigr)&\in  \caS_{n}\times \wp(\Lambda)^{n+1}\times \wp'(\Lambda)^n\,\bigl| \\
&S_k\setminus B_{k}=S_{k-1}\setminus B_{k},\,k=1,\dots,n \bigr\}\vphantom{\int}
\end{split}
\end{equation}
where $\wp(\Lambda)$ ($\wp'(\Lambda)$) is the power set of $\Lambda$ (without the empty set). We furnish these sets with the obvious structure of a measurable space $(\frS_n,\caF_n)$, where $\caF_n$ is the product of the Lebesgue measurable sets within the simplex and the discrete $\sigma$-algebra on the finite set $\wp(\La)^{2n+1}$. We define a finite complex measure $W_n$ on $\mathfrak{S}_n$ which is determined by a density with respect to the Lebesgue measure on $\caS_{n}$: still for $n\geq 1$ we set
\begin{equation}\label{eq: density}
\begin{split}
 &\rho_n(\frX)\;\mathrm{d}t_1\dots\mathrm{d}t_n\\
 &:=\Tr\bigl( \caP(\emptyset)Q(x) \bigr)^{-1} \,\Tr\bigl(\caP(S_0)Q(x) \bigl(\caP(S_n)\e^{-(\beta-t_n) H_0}\Upsilon({B_n})\bigr)\dots \\
 &\hphantom{:=\Tr\bigl( \caP(\emptyset)Q(x) \bigr)^{-1} \,} \dots \bigl(\caP(S_1)\e^{-(t_2-t_1)H_0}\Upsilon({B_1})\bigr) \e^{-t_1 H_0} \caP(S_0) \bigr)\;\mathrm{d}t_1\dots\mathrm{d}t_n
\end{split}
\end{equation}
where $\frX \in \frS_n$. Recall that $\Upsilon(B)=0$ whenever $B$ is not a connected set. Throughout the construction one can always restrict to those $\boldsymbol{B}_n$ consisting of connected sets.
The denominator does not vanish 
and the above density is well-defined by Assumption 2. Also note that the density would vanish if the condition on $\bsS_n$ and $\bsB_n$ in the definition of $\frS_n$ were not satisfied because 
\begin{equation}
\caP(S')\e^{-t H_0}\Upsilon(B) \caP(S)=0 \quad \textnormal{if}\quad S'\setminus B\neq  S\setminus B
\end{equation}
For $n=0$ we decide that $\caS_{0}=\wp(\Lambda)^0=\emptyset$ and the discrete measure on $\frS_0=\wp(\Lambda)$ to be determined through
\begin{equation}\label{eq: density0}
 W_0(\frX)=\left\{
		\begin{array}{ll}
			0, & S_0=\emptyset \\
			\rho_0(\frX):=\Tr\bigl( \caP(\emptyset)Q(x) \bigr)^{-1}\,\Tr\bigl( \caP(S_0)Q(x)\e^{-\beta H_0} \bigr), & \mathrm{else}
		\end{array}\right.
\end{equation}
The reason for manually removing the weight on $\frX=\emptyset\in\frS_0$, the \textit{empty diagram}, will become apparent in combinatorial constructions to come. Recall that
\begin{equation}
\Tr\bigl( \caP(\emptyset)Q(x)\e^{-\beta H_0} \bigr)=\Tr\bigl( \caP(\emptyset)Q(x)\bigr)
\end{equation}
and thus 
\begin{equation}\label{eq: NonNormalClassProj}
 \Tr\bigl(Q(x)\e^{-\beta H}\bigr)=  \Tr\bigl( \caP(\emptyset)Q(x) \bigr) \Bigl[1+\sum_{n= 0}^\infty W_n\bigl( \frS_n \bigr)\Bigr]
\end{equation}
The logarithm of the left-hand side will give the potential for the classical restriction. 

\subsubsection{Graphical representation and factorization into polymers}\label{sub: GraphRep}
\begin{figure}[h]
 \def\svgwidth{.9 \textwidth}
 \begingroup%
  \makeatletter%
  \providecommand\color[2][]{%
    \errmessage{(Inkscape) Color is used for the text in Inkscape, but the package 'color.sty' is not loaded}%
    \renewcommand\color[2][]{}%
  }%
  \providecommand\transparent[1]{%
    \errmessage{(Inkscape) Transparency is used (non-zero) for the text in Inkscape, but the package 'transparent.sty' is not loaded}%
    \renewcommand\transparent[1]{}%
  }%
  \providecommand\rotatebox[2]{#2}%
  \ifx\svgwidth\undefined%
    \setlength{\unitlength}{560.22265625bp}%
    \ifx\svgscale\undefined%
      \relax%
    \else%
      \setlength{\unitlength}{\unitlength * \real{\svgscale}}%
    \fi%
  \else%
    \setlength{\unitlength}{\svgwidth}%
  \fi%
  \global\let\svgwidth\undefined%
  \global\let\svgscale\undefined%
  \makeatother%
  \begin{picture}(1,0.59585573)%
    \put(0,0){\includegraphics[width=\unitlength]{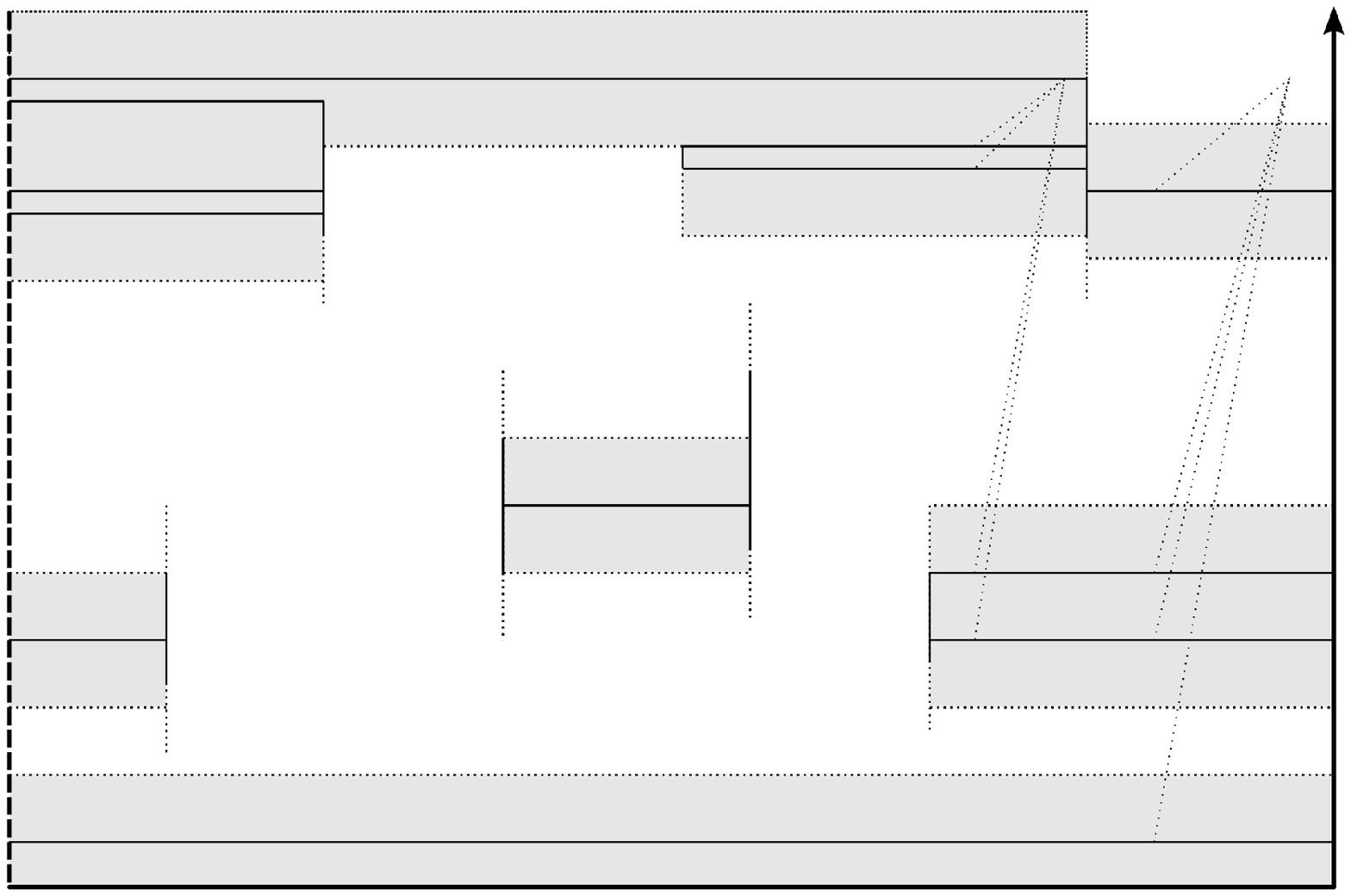}}%
    \put(0.88894064,0.52014336){\color[rgb]{0,0,0}\makebox(0,0)[lb]{\smash{$\Lambda$}}}%
    \put(0.63189998,0.00606204){\color[rgb]{0,0,0}\makebox(0,0)[lb]{\smash{$t_2$}}}%
    \put(0.74614028,0.00606204){\color[rgb]{0,0,0}\makebox(0,0)[lb]{\smash{$t_1$}}}%
    \put(0.50337966,0.00606204){\color[rgb]{0,0,0}\makebox(0,0)[lb]{\smash{$t_3$}}}%
    \put(0.46053955,0.00606204){\color[rgb]{0,0,0}\makebox(0,0)[lb]{\smash{$t_4$}}}%
    \put(0.11781867,0.00606204){\color[rgb]{0,0,0}\makebox(0,0)[lb]{\smash{$t_7$}}}%
    \put(0.23205896,0.00606204){\color[rgb]{0,0,0}\makebox(0,0)[lb]{\smash{$t_6$}}}%
    \put(0.33201922,0.00606204){\color[rgb]{0,0,0}\makebox(0,0)[lb]{\smash{$t_5$}}}%
    \put(0.63189998,0.16314244){\color[rgb]{0,0,0}\makebox(0,0)[lb]{\smash{$B_2$}}}%
    \put(0.73186024,0.4487432){\color[rgb]{0,0,0}\makebox(0,0)[lb]{\smash{$B_1$}}}%
    \put(0.51765969,0.24882266){\color[rgb]{0,0,0}\makebox(0,0)[lb]{\smash{$B_3$}}}%
    \put(0.47481958,0.50586332){\color[rgb]{0,0,0}\makebox(0,0)[lb]{\smash{$B_4$}}}%
    \put(0.14637874,0.16314244){\color[rgb]{0,0,0}\makebox(0,0)[lb]{\smash{$B_7$}}}%
    \put(0.246339,0.44874318){\color[rgb]{0,0,0}\makebox(0,0)[lb]{\smash{$B_6$}}}%
    \put(0.36057929,0.24882266){\color[rgb]{0,0,0}\makebox(0,0)[lb]{\smash{$B_5$}}}%
    \put(0.8318205,0.56298347){\color[rgb]{0,0,0}\makebox(0,0)[lb]{\smash{$S_0$}}}%
    \put(0.68902013,0.56298347){\color[rgb]{0,0,0}\makebox(0,0)[lb]{\smash{$S_1$}}}%
    \put(0.87466061,0.00606204){\color[rgb]{0,0,0}\makebox(0,0)[lb]{\smash{$0$}}}%
    \put(-0.00213364,0.00606204){\color[rgb]{0,0,0}\makebox(0,0)[lb]{\smash{$\beta$}}}%
  \end{picture}%
\endgroup%
 \caption{A sample diagram $\frX\in \frS_n$ for $n=7$ which is composed of $4$ polymers. \\ The fully-drawn horizontal segments correspond to the sets $S_0,\dots,S_n$. They are allowed to start or end only at fully drawn vertical segments, each corresponding to a (connected) \textit{interaction set} $B_i$ at time $t_i$, or at the boundaries of the diagram. }\label{fig: SampleDiagram}
\end{figure} 

As the name indicates we associate each $\frX \in \frS_n$ with a diagram `living on' $\Lambda\times [0,\beta]$, see Fig.~\ref{fig: SampleDiagram} (there $\Lambda\subset\bbZ$) for the details.

We refer to elements in $\La$ as \textit{spatial} points and to elements in $[0,\beta]$ as \textit{times}. The union of the vertical and horizontal segments constituting a diagram (fully drawn in Fig.~\ref{fig: SampleDiagram}) is denoted by
\begin{equation}
 \mathrm{Dom}(\frX):= \Bigl( \bigcup_{k=0}^n S_k \times [t_k,t_{k+1}] \Bigr) \cup \Bigl( \bigcup_{k=0}^n B_k \times t_k \Bigr) 
\end{equation}
where we set $t_0=0$, $t_{n+1}=\be$ and 
\begin{equation}
 \mathrm{Dom}_r(\frX):= \bigl\{ (z',t) \in \Lambda\times [0,\beta]\,\bigl|\, \textnormal{there is }(z,t)\in \mathrm{Dom}(\frX), \, \text{dist}(z,z')< r \bigr\}
\end{equation}
indicates the \textit{space-time volume} within the spatial interaction range $r$ to this domain (the shaded areas and dashed segments in  Fig.~\ref{fig: SampleDiagram}).\\
We say that two diagrams $\frX\in \frS_n$, $\frX'\in \frS_m$ are \emph{adjacent}, $\frX\leftrightarrow\frX'$, if  $\mathrm{Dom}_r(\frX)\cap\mathrm{Dom}_r(\frX')\neq\emptyset$; otherwise we write $\frX\nleftrightarrow\frX'$. A given diagram $\frX\in \frS_n$ is called a \textit{polymer} if there are no two diagrams $\frZ\in \frS_l$, $\frZ'\in \frS_m$, $(l+m)=n$, so that $\mathrm{Dom}(\frX)=\mathrm{Dom} (\frZ)\cup \mathrm{Dom}(\frZ')$ and $\frZ\nleftrightarrow\frZ'$. The set of polymers of $n$ interactions is denoted by $\frP_n$. Every diagram $\frX_n$ has a unique decomposition into polymers $\{\frp_\alpha\}$, $\frp_\alpha=(\bst_n^{\alpha},\bsS_n^{\alpha},\bsB_n^{\alpha}) \in \frP_{n(\alpha)}$, where $\alpha$ runs over a finite index set and $\sum _\alpha n(\alpha)=n$. For a given diagram $\frX \in \frS_n$, we denote with $\caR(\frX):=S_n\cup S_0$ its so-called \textit{root set}.

We have the following factorization and locality properties:
\begin{lemma}\label{lem}
 Let $\frX =(\bst_n,\bsB_n,\bsS_n)\in \frS_n$, $n\geq 0$, be a diagram with polymer decomposition $\{\frp_\alpha\}$, $\frp_\alpha \in \frP_{n(\alpha)}$, then the density factorizes according to
 \begin{equation}
  \rho_{n}^{}(\frX)=\prod_\alpha \rho_{n(\alpha)}^{}(\frp_\alpha)
 \end{equation}
 Moreover $\rho_{n}^{}(\frX)$ does not depend on the volume $\Lambda\subset\bbZ$, assuming of course that $\mathrm{Dom}_r(\frX)\subset \Lambda\times [0,\beta]$, and it is independent of $x_i\in\Omega_{i}$ whenever $i\notin \caR(\frX)$.
\end{lemma}

\begin{proof}
 In case of spatial separation of two diagrams the factorization property simply follows from the locality of the involved operators, the finite range $r$ of the unperturbed potential $\Phi_0$ and the fact that the trace of a tensor product of operators factorizes in the same way. For the other case of separation in time recall that $\caP$ is a one-dimensional projection, which eliminates non-commutativity of polymers separated in time. More precisely, let us choose a product basis of $\caH_A$, $A\subset\La$, denoted by
 \begin{equation}
 \bigl\{b(l_A)\bigr\}:=\bigl\{{\textstyle\bigotimes}_{i\in A}\, b_i(l_i)\bigr\}, \quad l_A=(l_i)_{i\in A}, \quad l_i=0,\dots,m-1 
 \end{equation}
 The unperturbed ground state is denoted by $b_i(0)=\caP_i b_i(0)$, then
 \begin{equation}
  \begin{split}
   &\rho_{n}(\frX)\\&=\prod_{i\in \caR(\frX)}\bigl\langle b_i(0)\,,\,Q_i(x_i) b_i(0) \bigr\rangle_i^{-1}\\    
   &\times \sum_{\substack{(l_i), i\in{S_0}\\l_i\neq 0}} \Bigl\langle b( l_{\La\setminus S_0}\equiv 0) \otimes  b(l_{S_0})\,,\, \bigl(\otimes_{i\in \caR(\frX)}Q_i(x_i)\otimes \opunit_{\Lambda\setminus \caR(\frX)}\bigr)\\[-6mm] &\hspace{.2\textwidth}\bigl(\caP(S_n)\e^{-(\beta-t_n) H_0}\Upsilon(B_n)\bigr)\dots \bigl(\caP(S_1)\e^{-(t_2-t_1)H_0}\Upsilon(B_1)\bigr) \e^{-t_1H_0}  \\ 
   &\hspace{.5 \textwidth}b( l_{\La\setminus S_0}\equiv 0) \otimes  b(l_{S_0}) \Bigr\rangle \\
   &=\prod_\alpha \prod_{i\in \caR(\frp(\al))}\bigl\langle b_i(0)\,,\,Q_i(x_i) b_i(0) \bigr\rangle_i^{-1}\\    
   &\times  \sum_{\substack{(l_i), i\in{S_0^{\al}}\\l_i\neq 0}} \Bigl\langle  b(l_{\La \setminus S_0^{\alpha}}\equiv 0) \otimes b(l_{S_0^{\alpha}})\,,\, \bigl(\otimes_{i\in \caR(\frp_\alpha)}Q_i(x_i)\otimes \opunit_{\Lambda\setminus \caR(\frp_\alpha)}\bigr)\\[-1.1\baselineskip] 
   &\hspace{.2\textwidth}\bigl(\caP(S_{n(\alpha)}^{\alpha})\e^{-(\beta-t_{n(\alpha)}^{\alpha}) H_0}\Upsilon(B_{n(\alpha)}^{\alpha})\bigr)\dots \bigl(\caP(S_1^{\alpha})\e^{-(t_2^{\alpha}-t_1^{\alpha})H_0}\Upsilon(B_1^{\alpha})\bigr) \vphantom{int}   \\ 
   &\hspace{.5 \textwidth}\e^{-t_1 H_0}\;b(l_{\La \setminus S_0^{\alpha}}\equiv 0) \otimes b(l_{S_0^{\alpha}})\vphantom{int} \Bigr\rangle
  \end{split}
 \end{equation}
The assertion that the density $\rho_n(\frX)$ is independent of the volume, in which the diagram $\frX$ is embedded, and of the local configuration $x_i$ at $i\notin \caR(\frX)$ is evident from the above expression.

\end{proof}

Let us denote with 
\begin{equation} 
 \left( \frP_\Lambda^\beta= \right)\;\frP:=\bigcup_{n\geq 0}\frP_n  
\end{equation}
the disjoint union of the set of polymers with $n$ interactions and furnish this set with the $\sigma$-algebra $\caF^\frp$ generated by $\bigcup_{n\geq 0}\caF_n^\frp$, where  $\caF_n^{\frp}$ is the $\sigma$-algebra induced by $\caF_n$ on the subset $\frP_n\subset \frS_n$.  
 If $\sum_{n\geq 1}|W_n|(\frP_n)<\infty$, $|W_n|$ the variation, as will be shown in Prop.~\ref{prop: KoteckyPreis}, then there is a complex measure $W$ on $(\frP, \caF^\frp)$ with finite total variation $|W|(\frP)<\infty$, such that $W=W_n$, on $\frP_n$. With this in mind we abbreviate $\sum_n\int\mathrm{d}W_n$ with $\int\mathrm{d}W$ from now on. Let $\chi[\cdot]$ denote the indicator function. Another consequence of the future Proposition~\ref{prop: KoteckyPreis} is that
 \begin{equation}\label{eq: summabilityW}
  \sum_{N\geq 0}\frac{1}{N!}\int_\frP \mathrm{d}|W|(\frp_1)\dots \int_\frP \mathrm{d}|W|(\frp_N)\prod_{1\leq i<j\leq N}\chi[\frp_i \nleftrightarrow \frp_j]<\infty
 \end{equation}
so that the factorization property of Lemma \ref{lem} allows to do the following reordering, here called \emph{polymer expansion},
\begin{equation}\label{eq: polymerexpansion}
 \sum_{n= 0}^\infty W_n(\frS_n)=\sum_{N= 1}^\infty\frac{1}{N!}\int_\frP \mathrm{d}W(\frp_1)\dots \int_\frP \mathrm{d}W(\frp_N)\prod_{1\leq i<j\leq N}\chi[\frp_i \nleftrightarrow \frp_j]
\end{equation}
On the left-hand side the sum is over diagrams with $n$ interactions and on the right-hand side over diagrams composed of $N$ polymers which are pair-wise non-adjacent. On the right-hand side the combined integration additionally includes diagrams where some of the times $\bst_n=(t_1,\dots,t_n)$ coincide, which is however only a contribution of measure zero. 

\subsection{Koteck\'y-Preis Criterion}\label{sec: KoteckyPreis}

In the following we prove a `Koteck\'y-Preis criterion' for our polymer model (Proposition \ref{prop: KoteckyPreis}), which allows to express~\eqref{eq: polymerexpansion} as an exponential of an integral over weighted \textit{clusters}, \ie sets of polymers which form connected graphs with respect to the graph structure given the adjacency relation `$\leftrightarrow$'. The underlying combinatorics go back to \cite{KP} and the generalization used here can be reviewed in \cite{Ue}.

\subsubsection{Decomposition of the polymers into constituents}

We decompose the polymers into \emph{constituents} that have a simpler structure. They can be seen as the vertices of yet another polymer model. \\
Denote by
\begin{equation}
 \frK:=(\wp'(\Lambda) \times [0,\beta])\cup \Lambda \equiv \frK_\mathrm{v}\cup\frK_\mathrm{h}
\end{equation}
the set of these constituents. Elements from the first part of the disjoint union may be thought of as the (connected) vertical segments in our diagrammatic description and elements from $\frK_{\mathrm{h}}$ are represented by horizontal segments at $i \in \Lambda$, which connect both boundaries, see also Fig.~\ref{fig: PolConstit}.  In this sense we define the (extended) domain of constituents $\frx\in\frK$ by
\begin{equation}
 \mathrm{Dom}_r(\frx):=\left\{\begin{array}{ll}
                \bigl\{(z',t)\,\bigl|\,\text{dist}(z',z)<r,\,z\in B  \bigr\}&\;\mathrm{if}\; \frx=(B,t)\in\frK_{\mathrm{v}} \\[1mm]
                \bigl\{ z'\in\Lambda\,\bigl|\,\text{dist}(i,z')<r \bigr\}\times[0,\beta] &\;\mathrm{if}\;\frx=i\in\frK_{\mathrm{h}}
               \end{array}\right. 
\end{equation}
For a given polymer $\frp\in \frP$ and constituent $\frx \in\frK$ we write $\frp\leftrightarrow \frx$ if and only if $\mathrm{Dom}_r(\frp)\cap\mathrm{Dom}_r(\frx)\neq \emptyset$, otherwise we write $\frp\nleftrightarrow \frx$.\\
We fix two positive constants $\alpha_1,\alpha_2>0$ and construct a measure $w$ on $\frK$ (not to be confused with the weights $w$ appearing in the high-temperature section):
\begin{equation}
\begin{split}
 w(B,\mathrm{d}t)&=4^{|B|}\cdot\exp((\alpha_2+\gamma) |B|)\lVert \Upsilon(B) \rVert\,\mathrm{d}t,\quad \text{on }  \frK_{\mathrm{v}}\\
 w(i)&=\exp(-(g-\alpha_1)\beta+\gamma),\quad \text{on }  \frK_{\mathrm{h}}
 \end{split}
\end{equation}
Here we defined
 \begin{equation}\label{eq: ObsNonOrth}
  \ga:=\max_{x\in \sp(X)}\,\log\biggl(\frac{m}{\Tr\bigl(Q(x) \caP\bigr)}\biggr)
 \end{equation}
and we may occasionally abuse notation and write $w(B,\mathrm{d}t)=w(B)\mathrm{d}t$, such that the expression $w(\frx)$ is meaningful for all $\frx\in\frK$. We also define the following functions on the set of diagrams, 
\begin{equation}
\begin{split}
a&:\frS\rightarrow \bbR\quad;\quad 	\frX \mapsto \alpha_1 L_h(\frX)+\alpha_2 L_v(\frX),\\
L_h&:\frS\rightarrow \bbR\quad;\quad 		\frX \mapsto	\sum_{i=0}^{n}|S_i|\cdot|t_{i+1}-t_i|\\
L_v&:\frS\rightarrow \bbR\quad;\quad 		\frX \mapsto	 \sum_{i=1}^{n}|B_i|
\end{split}
\end{equation}
where as before $\frX=(\bst_n,\bsB_n,\bsS_n)$, $n\geq 1$, and $L_h$/$L_v$ indicate the corresponding diagram's total \emph{length} of the horizontal/vertical segments. 
We furthermore introduce a symmetric function $\xi_{\al_1}:\frK\times \frK\rightarrow \bbR$,
\begin{equation}
 \xi_{\alpha_1}(\frx,\frx')=\left\{
		\begin{array}{ll}
			\e^{-(g-\alpha_1)|t-t'|} &\;\textnormal{if}\;\frx,\frx'\in \frK_{\mathrm{v}},\, \text{dist}(B,B')<2r \\[1mm]
			1 &\;\textnormal{if}\;\mathrm{Dom}_r(\frx)\cap\mathrm{Dom}_r(\frx')\neq\emptyset,\,\textnormal{and not both } \frx,\frx'\in \frK_{\mathrm{v}}\\[1mm]
			0 &\;\textnormal{else}
		\end{array}
	\right.
\end{equation}
 \begin{figure}[h]
 \def\svgwidth{1.1 \textwidth}
  \begingroup%
  \makeatletter%
  \providecommand\color[2][]{%
    \errmessage{(Inkscape) Color is used for the text in Inkscape, but the package 'color.sty' is not loaded}%
    \renewcommand\color[2][]{}%
  }%
  \providecommand\transparent[1]{%
    \errmessage{(Inkscape) Transparency is used (non-zero) for the text in Inkscape, but the package 'transparent.sty' is not loaded}%
    \renewcommand\transparent[1]{}%
  }%
  \providecommand\rotatebox[2]{#2}%
  \ifx\svgwidth\undefined%
    \setlength{\unitlength}{733.5796875bp}%
    \ifx\svgscale\undefined%
      \relax%
    \else%
      \setlength{\unitlength}{\unitlength * \real{\svgscale}}%
    \fi%
  \else%
    \setlength{\unitlength}{\svgwidth}%
  \fi%
  \global\let\svgwidth\undefined%
  \global\let\svgscale\undefined%
  \makeatother%
  \begin{picture}(1,0.29880871)%
    \put(0,0){\includegraphics[width=\unitlength]{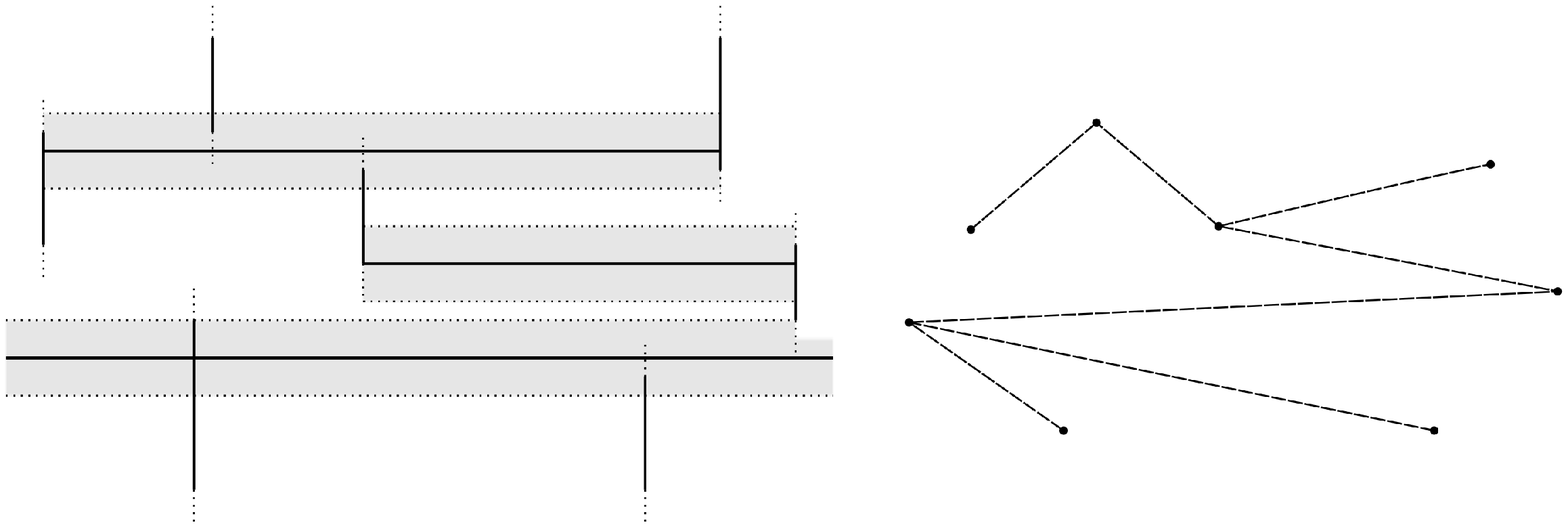}}%
    \put(0.10223093,0.72662703){\color[rgb]{0,0,0}\makebox(0,0)[lt]{\begin{minipage}{0.02414773\unitlength}\raggedright am\end{minipage}}}%
    \put(0.1404048,0.22897811){\color[rgb]{0,0,0}\makebox(0,0)[lt]{\begin{minipage}{0.02258981\unitlength}\raggedright 1\end{minipage}}}%
    \put(0.43486107,0.20716239){\color[rgb]{0,0,0}\makebox(0,0)[lt]{\begin{minipage}{0.0264846\unitlength}\raggedright 2\\ \end{minipage}}}%
    \put(0.04225108,0.16353582){\color[rgb]{0,0,0}\makebox(0,0)[lt]{\begin{minipage}{0.03505315\unitlength}\raggedright 3\end{minipage}}}%
    \put(0.22762873,0.16353582){\color[rgb]{0,0,0}\makebox(0,0)[lt]{\begin{minipage}{0.03271628\unitlength}\raggedright 4\end{minipage}}}%
    \put(0.47845356,0.11989463){\color[rgb]{0,0,0}\makebox(0,0)[lt]{\begin{minipage}{0.02804254\unitlength}\raggedright 5\end{minipage}}}%
    \put(0.12950907,0.02177013){\color[rgb]{0,0,0}\makebox(0,0)[lt]{\begin{minipage}{0.02960047\unitlength}\raggedright 7\end{minipage}}}%
    \put(0.39121989,0.02174578){\color[rgb]{0,0,0}\makebox(0,0)[lt]{\begin{minipage}{0.02882151\unitlength}\raggedright 8\end{minipage}}}%
    \put(-0.00136007,0.09811545){\color[rgb]{0,0,0}\makebox(0,0)[lt]{\begin{minipage}{0.02643875\unitlength}\raggedright 6\\ \end{minipage}}}%
    \put(0.64205359,0.25274527){\color[rgb]{0,0,0}\makebox(0,0)[lt]{\begin{minipage}{0.01855591\unitlength}\raggedright 1\end{minipage}}}%
    \put(0.88198098,0.22002499){\color[rgb]{0,0,0}\makebox(0,0)[lt]{\begin{minipage}{0.02175521\unitlength}\raggedright 2\\ \end{minipage}}}%
    \put(0.56025887,0.18185199){\color[rgb]{0,0,0}\makebox(0,0)[lt]{\begin{minipage}{0.02879366\unitlength}\raggedright 3\end{minipage}}}%
    \put(0.70747014,0.17094657){\color[rgb]{0,0,0}\makebox(0,0)[lt]{\begin{minipage}{0.02687409\unitlength}\raggedright 4\end{minipage}}}%
    \put(0.92012598,0.143667){\color[rgb]{0,0,0}\makebox(0,0)[lt]{\begin{minipage}{0.02303494\unitlength}\raggedright 5\end{minipage}}}%
    \put(0.63115613,0.0618963){\color[rgb]{0,0,0}\makebox(0,0)[lt]{\begin{minipage}{0.02431468\unitlength}\raggedright 7\end{minipage}}}%
    \put(0.84597708,0.05969522){\color[rgb]{0,0,0}\makebox(0,0)[lt]{\begin{minipage}{0.02367481\unitlength}\raggedright 8\end{minipage}}}%
    \put(0.52209855,0.12733488){\color[rgb]{0,0,0}\makebox(0,0)[lt]{\begin{minipage}{0.02171754\unitlength}\raggedright 6\\ \end{minipage}}}%
    \put(0.59336464,0.08281735){\color[rgb]{0,0,0}\makebox(0,0)[lb]{\smash{1}}}%
    \put(0.73513516,0.08281735){\color[rgb]{0,0,0}\makebox(0,0)[lb]{\smash{1}}}%
    \put(0.6697026,0.12643906){\color[rgb]{0,0,0}\makebox(0,0)[lb]{\smash{1}}}%
    \put(0.68021858,0.20939818){\color[rgb]{0,0,0}\makebox(0,0)[lb]{\smash{$\e^{-g|t_4-t_1|}$}}}%
    \put(0.59297516,0.17122919){\color[rgb]{0,0,0}\makebox(0,0)[lb]{\smash{$\e^{-g|t_3-t_1|}$}}}%
    \put(0.75110386,0.19304004){\color[rgb]{0,0,0}\makebox(0,0)[lb]{\smash{$\e^{-g|t_4-t_2|}$}}}%
    \put(0.81653642,0.15487105){\color[rgb]{0,0,0}\makebox(0,0)[lb]{\smash{$\e^{-g|t_5-t_4|}$}}}%
  \end{picture}%
\endgroup%
  \caption{A polymer for a given skeleton. On the right the graph on its constituents from $\frK=\frK_{\mathrm{v}}\cup \frK_{\mathrm{h}}$ is depicted; $\frK_{\mathrm{v}}$ contains the vertical segments, $\frK_{\mathrm{h}}$ the end-to-end horizontal segments.}\label{fig: PolConstit}
 \end{figure}
 
The relevant relationship between the constituents and the polymers, which they compose, is provided by the following lemma. All following results hold uniformly in $\La\Subset \bbZ^d$ and $x\in\Om$.\\
For a polymer $\frp=(\bst_n,\bsB_n,\bsS_n)\in\frP_n$ we define $I(\frp):=\bigcap_{i=0}^n S_i$, which indicates the location of end-to-end horizontal segments in the diagrams, and we define its \emph{skeleton} 
\begin{equation}
\skl(\frp)=\{(B_1,t_1),\dots,(B_n,t_n)\}\cup I(\frp)\subset \frK
\end{equation}
which is a collection of constituents. With $\skl(\frP)$ we denote the set of all skeletons. On the other hand for a given skeleton $\frs\in\skl(\frP)$, we denote with $\frP(\frs)\subset\frP$ the set of those polymers $\frp$ whose skeleton $\skl(\frp)=\frs$. 

\begin{lemma}\label{lem: BigToSmall}
 Let $\alpha_1<g$ and $\alpha_2>0$, then for every skeleton $\frs=\{\frx_1,\dots,\frx_{N}\}\in \skl(\frP)$, we have
 \begin{equation}
  \begin{split}
   \sum_{\frp\in\frP(\frs)} \lvert\rho(\frp)\rvert \,\e^{a(\frp)} \leq \left(\prod_{i=1}^{N} w(\frx_i)\right) \max_{T\in\caT_{N}^*}\prod_{\{i,j\}\in E(T)} \xi_{\alpha_1}(\frx_i,\frx_j)
  \end{split}
 \end{equation}
where $\caT_{N}^*$ denotes the set of connected trees on the vertices $1,\dots,N$ and $E(T)$ denotes the set of edges in a connected tree $T$.
\end{lemma}

\begin{proof}
 Recall the details in the definition of the measure $W_n$ given in~\eqref{eq: density} and~\eqref{eq: density0}. Given $\frp=(\bst_n,\bsB_n,\bsS_n)\in\frP$, the density is bounded from above by 
 \begin{equation}\label{eq: BoundMeasure}
  \begin{split}
   &\lvert \rho(\frp)\rvert\\
   &\leq\Bigl\lvert\Bigl(\prod_{i\in\caR(\frp)}\Tr_i\bigl(Q_i(x_i)\caP_i\bigr)^{-1}\Bigr) \Tr\Bigl( \bigl(\otimes_{i\in \caR(\frp)} Q_i(x_i)\bigr)\otimes\bigl( \otimes_{i\notin \caR(\frp)} \caP_i \bigr) \\[-.8\baselineskip]
   &\hspace{.15\textwidth} \bigl(\caP(S_n)\e^{-(\beta-t_n) H_0}\Upsilon(B_n)\bigr)\dots \bigl(\caP(S_1)\e^{-(t_2-t_1)H_0}\Upsilon(B_1)\bigr) \caP(S_0) \e^{-t_1 H_0}\Bigr)\vphantom{\int}\Bigr\rvert\\
   & \leq \e^{\gamma(|I(\frp)|+L_v(\frp))}\; \bigl\lVert \bigl(\caP(S_n)\e^{-(\beta-t_n) H_0}\Upsilon(B_n)\bigr)\dots\\
   &\hspace{.3 \textwidth}\dots\bigl(\caP(S_1)\e^{-(t_2-t_1)H_0}\Upsilon(B_1)\bigr) \caP(S_0) \e^{-t_1 H_0} \bigr\rVert\vphantom{\int}\\
   &\leq \e^{\gamma(|I(\frp)|+L_v(\frp))}_{}\e^{-gL_h(\frp)}_{}\prod_{k=1}^n \lVert \Upsilon({B_k}) \rVert
  \end{split}
 \end{equation}
 where the second trace was estimated by the norm of the operator product (note that the norm of each orthogonal projections $\caQ_i$ and $\caP_i$ equals one) multiplied with $\rank (\caP(S_0))<m^{\lvert \frR(\frp)\rvert}$, $m=\dim(\caH)$. This factor together with the product over inverse traces then was absorbed in $\e^{\gamma(|I(\frp)|+L_v(\frp))}$ where we used that $\lvert \frR(\frp)\rvert\leq|I(\frp)|+L_v(\frp)$. The last inequality in \eqref{eq: BoundMeasure} follows from the Peierls' type condition~\eqref{eq: DecayOfExcitations}.\\
 Take any skeleton and cast it in the form
 \begin{equation}\label{eq: SkeletonExplicit}
  \begin{split}
   \frs=\frs_{n,k}&=\{\frx_1,\dots,\frx_n,\frx_{n+1}, \dots,\frx_N\}\in\skl(\frP_n)\\
   &\frx_i=(B_i,t_i)\in\frK_{\mathrm{v}},\quad i=1,\dots,n\\
   &\frx_i\in I\subset\frK_{\mathrm{h}},\quad i=n+1,\dots,(n+k)= N,\quad k:=\lvert I\rvert
  \end{split}
 \end{equation}
 The minimal horizontal length of any polymer $\frp\in\frP(\frs)$ belonging to such a skeleton can be estimated as
 \begin{equation}
  L_{\mathrm{h}}(\frp) \geq \lvert I \rvert \beta +\min_{T\in\caT_N^*}\chi\bigl[\xi_{\al_1}(\frx_i,\frx_j)\neq 0\;\forall\, \{i,j\}\in E(T)\bigr] \sum_{\substack{\{i,j\}\in E(T);\\i,j\leq n}} \lvert t_i - t_j \rvert
 \end{equation}
 The first term accounts for the contribution from end-to-end segements. The second term gives the minimal length of the horizontal segments which, diagrammatically speaking, must be added to the skeleton between vertical constituents to obtain a polymer. Therefore, since $\al_1<g$,
 \begin{equation}
  \e_{}^{-(g-\al_1)L_{\mathrm{h}}(\frp)}\leq \e_{}^{-(g-\al_1)\lvert I \rvert} \max_{T\in\caT_{N}^*}\prod_{\{i,j\}\in E(T)} \xi_{\alpha_1}(\frx_i,\frx_j)
 \end{equation}
 and using the bound \eqref{eq: ObsNonOrth} gives
 \begin{equation}
  \begin{split}
   \e^{a(\frp)}_{} \,\lvert \rho(\frp)\rvert\leq &  \left( \prod_{i=1}^{n}\e^{(\al_2+\ga)\lvert    B_i\rvert}\bigl\lVert \Upsilon(B_i)\bigr\rVert \right)\,\e_{}^{-(g-\al_1)\be\lvert I\rvert+\ga\lvert I\rvert}\\
    &\times \max_{T\in\caT_{N}^*}\prod_{\{i,j\}\in E(T)} \xi_{\alpha_1}(\frx_i,\frx_j)
   \end{split}
 \end{equation}
 for any polymer $\frp\in\frP(\frs)$. The lemma then follows from the fact that $\lvert \frP(\frs)\rvert\leq  4^{L_{\mathrm{v}}(\frp)}$, where $L_{\mathrm{v}}(\frp)$ is of course independent of the choice $\frp\in\frP(\frs)$.\\
\end{proof}

The next lemma is concerned only with the constituent model. It gives a bound on the integral over `clusters of constituents', where the word cluster here refers to a collection of constituents viewed as vertices that is a connected graph w.r.~to $\xi_{\al_1}$ viewed as edge weight.

\begin{lemma}\label{lem: EstConst}
 For all $\alpha_1, \alpha_2$ with $\alpha_1<g$ and $0<\de_1,\de_2<1$ there exist $\kappa_{\min},\beta_{\min}>0$, so that, for any $\kappa\geq\kappa_{\min}$, $\beta\geq\beta_{\min}$ and for every constituent $\frx_0\in\frK$,
 \begin{equation}\label{eq: ConstModel}
 \begin{split}
  1+\sum_{N= 1}^\infty \frac{1}{N!} \int_\frK \mathrm{d}w(\frx_1)\dots\int_\frK \mathrm{d}w(\frx_N) \max_{T\in \caT_{N}}\prod_{\{i,j\}\in E(T)}\xi_{\alpha_1}(\frx_i,\frx_j)  \leq \e^{d(\frx_0)}
  \end{split}
 \end{equation}
 where 
 \begin{equation}
  d(\frx_0):=\left\{ \begin{array}{ll} 
   \de_1(g-\alpha_1) \beta &\; \text{if }\frx_0=i\in\frK_\mathrm{h}\\
   \de_2|B| &\; \textnormal{if }\frx_0=(B,t)\in \frK_\mathrm{v}
  \end{array} \right.
\end{equation}
 and where $\caT_N$ denotes the set of all connected trees on the vertices $0,\dots,N$.
\end{lemma}

\begin{proof}
 In fact, we prove a stronger version of the lemma by replacing the above maximum by a sum over all connected trees.  We truncate the series, \ie, replace $\sum_{N=1}^\infty$ by $\sum_{N=1}^M$, and then proceed by induction on $M$. By the exponential decay of the perturbation interaction $\lVert \Upsilon \rVert_\kappa\leq 1$ and by counting the possible constituents $\frx_1$ that can be attached to the fixed one $\frx_0$, \ie, with $\xi(\frx_0,\frx_1)\neq 0$, it is not hard to see that, for sufficiently large $\kappa$ and $\beta$, one has the bound:
 \begin{equation}
  \begin{split}
   &\int_\frK \mathrm{d}w(\frx_1)\,\xi_{\alpha_1}(\frx_0,\frx_1)\,\e^{d(\frx_1)}\\
   &\leq C(\al_1,\al_2)\times\left\{ \begin{array}{ll} 
               \e^{-(1-\de_1)(g-\alpha_1)\beta}+   \e^{-\kappa}\,\beta  &\; \textnormal{if }\frx_0=i\in\frK_\mathrm{h}\\[1mm]
               e^{-(1-\de_1)(g-\al_1)\beta}\,|B|  + \e^{-\kappa}\,|B| &\; \textnormal{if }\frx_0=(B,t)\in \frK_\mathrm{v}
             \end{array} \right.
  \end{split}
 \end{equation}
 where $C(\al_1,\al_2)$ is an irrelevant constant depending only on $\al_1$, $\al_2$. This bound immediately implies that
 \begin{equation}\label{eq: IndStartStrong}
 \int_\frK \mathrm{d}w(\frx_1)\,\xi_{\al_1}(\frx_0,\frx_1)\,\e^{d(\frx_1)}
 \leq d(\frx_0)
\end{equation}
 for $\be, \ka$ large enough. It also allows to start the induction at \underline{$M=1$}. \\
To obtain the induction step \underline{$M-1\rightarrow M$} we first sort the terms within the sum over trees by the number $m$ of different constituents, say $\frx_{1}$, that are connected to $\frx_0$ in the sense $\xi_{\al_1}(\frx_0,\frx_{1})\neq 0$. Each $\frx_{1}$ is itself connected to at most $M-m$ other constituents, so that the induction hypothesis can be used. 
\begin{equation}
 \begin{split}
  &\sum_{N=1}^M\frac{1}{N!}\sum_{T\in\caT_{N}}\int_\frK \mathrm{d}w(\frx_1)\dots\int_\frK \mathrm{d}w(\frx_N)\,\prod_{\{i,j\}\in E(T)}\xi_{\al_1}(\frx_i,\frx_j)\\
  &\leq \sum_{m=1}^M\frac{1}{m!} \Bigl[   \int_\frK \mathrm{d}w(\frx_1)\,\xi_{\al_1}(\frx_0,\frx_1)\\
  &\hspace{.1\textwidth}\sum_{N=0}^{M-m}\frac{1}{N!}\sum_{T\in \caT_{N+1}^*} \int_\frK \mathrm{d}w(\frx_2)\dots\int_\frK \mathrm{d}w(\frx_{N+1})\,\prod_{\{i,j\}\in E(T)}\xi_{\al_1}(\frx_i,\frx_j) \Bigr]^m\\
  &\leq \sum_{m=1}^M\frac{1}{m!} \Bigl[   \int_\frK \mathrm{d}w(\frx_1)\,\xi_{\al_1}(\frx_0,\frx_1) \e^{d(\frx_1)} \Bigr]^m\leq \e^{d(\frx_0)}-1
 \end{split}
\end{equation}
where $T_N^*$ again denotes the set of connected trees on the vertices $1,\dots,N$ and the $N=0$ term in the sum is again understood to be equal to one. For the last inequality we used \eqref{eq: IndStartStrong}.

\end{proof}

\subsubsection{Koteck\'y--Preis criterion}
Now we prove a Koteck\'y--Preis type criterion for our polymer model which is an upper bound for the integral over polymers which are adjacent to a fixed polymer $\frp_0\in\frP$.

\begin{proposition}\label{prop: KoteckyPreis}
 For all $\alpha_1,\alpha_2>0$ with $\alpha_1<g$ and constants $c_1,c_2>0$, there exist $\kappa_{\min},\beta_{\min}>0$, such that, for all $\kappa\leq\kappa_{\min}$, $\beta\leq\beta_{\min}$,
 \begin{equation}\label{eq: kotpre}\begin{split}
  &\int_{\frP} \mathrm{d}|W|(\frp)\;\chi[\frp\leftrightarrow\frp']\, \e^{a(\frp)}\leq c_1 L_h(\frp')+c_2 L_v(\frp')\\
  \text{and}\quad &\int_{\frP} \mathrm{d}|W|(\frp)\;\e^{a(\frp)}<\infty
 \end{split}\end{equation}
for every fixed polymer $\frp'\in \frP$, volume $\Lambda\Subset\bbZ^d$, and classical configuration $x\in\Omega$.
\end{proposition}

\begin{proof}
Note that $\frp\leftrightarrow \frp'$ implies at least one of the following conditions:
\begin{itemize}
 \item[(i)] The vertical skeleton of $\frp$ is `connected' to $\frp'$, \ie, \\
 $\exists \,\frx \in \skl (\frp)\cap \frK_\mathrm{v} \text{ such that } \frx\leftrightarrow\frp'$
 \item[(ii)] The horizontal skeleton of $\frp$ is `connected' to $\frp'$, \ie,\\ 
 $\exists \,\frx \in \skl (\frp)\cap \frK_\mathrm{h} \text{ such that } \frx\leftrightarrow\frp'$
 \item[(iii)] A horizontal segment of $\frp$ that is not end-to-end is `connected' to $\frp'$, \ie,\\
 $\bigl[ \mathrm{Dom}_r(\frp)\setminus {\textstyle\bigcup_{\frx\in\skl(\frp)}}\mathrm{Dom}_r(\frx) \bigr]\cap\mathrm{Dom}_r(\frp')\neq\emptyset$
\end{itemize}
Furthermore it can be seen that polymers $\frp, \frp'\in\frP$ for which (iii) holds must satisfy either (i) and/or 
\begin{itemize}
 \item[(iii')] $\frp$ is `connected' to the vertical skeleton of $\frp'$, \ie,\\
 $\exists \,\frx' \in \skl(\frp')\cap \frK_{\mathrm{v}} \text{ such that } \frp\leftrightarrow\frx'$
\end{itemize}
At last, note that polymers $\frp, \frp'\in\frP$ for which (ii) is true must satisfy either (iii') and/or
\begin{itemize}
 \item[(ii')] The horizontal skeletons of $\frp$ and $\frp'$ are connected, \ie,\\
 $\exists \,\frx \in \skl (\frp)\cap \frK_\mathrm{h},\,\frx' \in \skl (\frp')\cap \frK_\mathrm{h} \text{ such that } \xi_{\alpha_1}(\frx,\frx')=1$
\end{itemize}
Therefore
\begin{equation}
\begin{split}
  &\int_{\frP} \mathrm{d}|W|(\frp)\;\chi[\frp\leftrightarrow\frp_0]\, \e^{a(\frp)}\\
  &\leq \int_{\frP} \mathrm{d}|W|(\frp)\;\bigl(\chi[\text{(i)}]+\chi[\text{(ii')}]+\chi[\text{(iii')}]\bigr)\, \e^{a(\frp)}
 \end{split}\end{equation}
and we proceed by giving bounds for each of the three terms. 

For the case (i) we first reorganize the integral for given $n,k\geq 0$ and $\bst_n\in\caS_n$ by collecting polymers with common skeleton of the form $\frs_{n,k}=\{\frx_1,\dots,\frx_{n+k}\}$ parametrized as in \eqref{eq: SkeletonExplicit} with $\frx_j\equiv(B_j,t_j)\in\frK_\mathrm{v}$, $j=1,\dots,n$, and where $\frx_{j+n}\equiv i_{j}\in\frK_{\mathrm{h}}$, for $j=1,\dots,k$, enumerates elements in $I\subset\frK_\mathrm{h}$ (in arbitrary order). 
\begin{align}\label{eq: CaseOne}\nonumber
  &\int_{\frP} \mathrm{d}|W|(\frp)\;\chi[\text{(i)}]\, \e^{a(\frp)}\\ \nonumber
  &=\sum_{n=1}^{\infty}\sum_{k=0}^\infty \int_{\caS_n}\d t_1\dots\d t_n \sum_{\frs_{n,k}} \sum_{l=1}^n \chi[\frx_l\leftrightarrow \frp'] \sum_{\frp\in\frP(\frs_{n,k})} \e^{a(\frp)} \lvert \rho(\frp) \rvert\\\nonumber
  &\leq \sum_{n=1}^{\infty}\sum_{k=0}^\infty \int_{\caS_n}\d t_1\dots\d t_n \sum_{\frs_{n,k}} \sum_{l=1}^n \chi[\frx_l\leftrightarrow \frp'] \biggl(\prod_{i=1}^{N} w(\frx_i)\biggr) \max_{T\in\caT_{N}^*}\prod_{\{i,j\}\in E(T)} \xi_{\alpha_1}(\frx_i,\frx_j)\\
  &\leq \sum_{n=1}^{\infty}\frac{n}{n!}\sum_{k=0}^\infty \frac{1}{k!} \int_{[0,\be]^n}\d t_1\dots\d t_n \sum_{\substack{\bsB_n\\\in{\wp'(\La)}^{n}}}\chi[\frx_1\leftrightarrow \frp']\sum_{\substack{(\frx_{n+1},\dots,\frx_{n+k})\\\in\frH_\mathrm{h}^k}}\\\nonumber
  &\hspace{.4 \textwidth}\biggl(\prod_{i=1}^{N} w(\frx_i)\biggr) \max_{T\in\caT_{N}^*}\prod_{\{i,j\}\in E(T)} \xi_{\alpha_1}(\frx_i,\frx_j)\\\nonumber
  &\leq \int_{\frK_{\mathrm{v}}} \d w(\frx_1)\,\chi[\frx_1\leftrightarrow\frp']\\\nonumber
  &\hspace{.05\textwidth}\Bigl[1+\sum_{N=2}^\infty { \frac{1}{(N-1)!}} \int_\frK \mathrm{d}w(\frx_2)\dots\int_\frK \mathrm{d}w(\frx_N) \max_{T\in \caT^*_{N}}\prod_{\{i,j\}\in E(T)}\xi_{\alpha_1}(\frx_i,\frx_j)\Bigr]\\\nonumber
  &\leq  \int_{\frK_{\mathrm{v}}} \d w(\frx_1)\,\chi[\frx_1\leftrightarrow\frp']\,\e^{d(\frx_1)}_{}\\\nonumber
  &\leq C'\,\e^{-\ka}\,L_{\mathrm{h}}(\frp')
\end{align}
for a constant $C'$ which only depends on $\al_1,\al_2$. To obtain the first inequality Lemma \ref{lem: BigToSmall} was used. The integrand is explicitly invariant under exchange of time coordinates and for the second inequality we replaced the integration over the simplex $\caS_n$ by integrating the cube $[0,\be]^n$ and dividing by $n!$. Furthermore we spelled out the sum over skeletons $\frs_{n,k}$ more explicitly, but instead of summing over sets containing $k$ horizontal constituents we summed over $k$-tuples divided by $k!$ for the upper bound. The additional factor $n$ is a consequence of rewriting (made possible by the symmetrization) the condition that at least one vertical constituent, namely $\frx_1$, of the polymer $\frp$ must be adjacent to $\frp'$. One arrives at the third inequality by taking out the integral over this adjacent constituent $\frx_1$ and by writing the remaining sums and integrals as multiple integral over (both horizontal and vertical) constituents. The last two steps follow 
from Lemma \ref{lem: EstConst} with $\frx_1$ assuming the role of $\frx_0$ in the Lemma. \\
Similarly we proceed in case (ii'), but this time in the sum over $\frs_{n,k}=\{\frx_1,\dots,\frx_{k+n}\}$ we first enumerate horizontal constituents, \ie, $\frx_j\equiv i_j\in\frK_\mathrm{h}$ for $j=1,\dots,k$ (again in arbitrary order), and then vertical ones, \ie, $\frx_{j+k}\equiv(B_j,t_j)\in\frK_\mathrm{v}$ for $j=1,\dots,n$. We get the following upper bound
\begin{equation}\label{eq: CaseTwo}
 \begin{split}
 &\int_{\frP} \mathrm{d}|W|(\frp)\;\chi[\text{(ii')}]\, \e^{a(\frp)}\\
 &=\sum_{k=1}^\infty  \sum_{n=0}^{\infty}\int_{\caS_n}\d t_{1}\dots\d t_{n} \sum_{\frs_{n,k}}\sum_{l=1}^k
 \chi\Bigl[ \begin{array}{c} \exists\, \frx' \in \skl(\frp')\cap\frK_{\mathrm{h}};\\ \xi_{\al_1}(\frx_{l},\frx') =1\end{array}\Bigr]
 \sum_{\substack{\frp\in\\ \frP(\frs_{n,k})}}\e^{a(\frp)}\,\lvert\rho(\frp)\rvert\\
  &\leq\sum_{k=0}^\infty \frac{k}{k!}\sum_{n=1}^\infty\frac{1}{n!}\int_{[0,\be]^n}\d t_{1}\dots\d t_{n} \sum_{\substack{\bsB_n\\\in{\wp'(\La)}^{n}}}\sum_{\substack{(\frx_{n+1},\dots,\frx_{n+k})\\\in\frH_\mathrm{h}^k}}\biggl(\prod_{i=1}^{N} w(\frx_i)\biggr)\\
  &\hspace{.25\textwidth}\chi\Bigl[ \begin{array}{c} \exists\, \frx' \in \skl(\frp')\cap\frK_{\mathrm{h}};\\ \xi_{\al_1}(\frx_{1},\frx') =1\end{array}\Bigr]\max_{T\in\caT_{n+k}^*}\prod_{\{i,j\}\in E(T)}\xi_{\al_1}(\frx_i,\frx_j)\\
  &\leq \int_{\frK_{\mathrm{h}}}\d w(\frx_1)\,\chi\Bigl[ \begin{array}{c} \exists\, \frx' \in \skl(\frp')\cap\frK_{\mathrm{h}};\\ \xi_{\al_1}(\frx_{1},\frx') =1\end{array}\Bigr]\\
  &\hspace{.05\textwidth}\Bigl[1+\sum_{N=2}^\infty { \frac{1}{(N-1)!}} \int_\frK \mathrm{d}w(\frx_2)\dots\int_\frK \mathrm{d}w(\frx_N) \max_{T\in \caT^*_{N}}\prod_{\{i,j\}\in E(T)}\xi_{\alpha_1}(\frx_i,\frx_j)\Bigr]\\
  &\leq \int_{\frK_{\mathrm{h}}}\d w(\frx_1)\,\chi\Bigl[ \begin{array}{c} \exists\, \frx' \in \skl(\frp')\cap\frK_{\mathrm{h}};\\ \xi_{\al_1}(\frx_{1},\frx') =1\end{array}\Bigr] \,\e^{d(\frx_1)}\\
  &\leq C''\,\e^{-(1-\de_1)(g-\al_1)\be}\,\frac{L_{\mathrm{h}}(\frp')}{\be}
 \end{split}
\end{equation}
where the last fraction is a bound on the number of end-to-end segments in the polymer $\frp'$ and where $C''$ is another constant which only depends on $\al_1,\al_2$. $\de_1$ here has the same meaning as in Lemma \ref{lem: EstConst} and can be chosen to be small.\\
For the remaining third integral (iii') we first split the vertical skeleton of $\frp'$ into `singletons' $\frx_0\in\{(i,t)\in\frK_\mathrm{v}\,|\,i\in\La, \,\exists \,(B,t)\in\skl(\frp')\cap\frK_{\mathrm{v}};\,i\in B \}$ to obtain
\begin{equation}\label{eq: SplittingSingletons}
 \int_{\frP} \mathrm{d}|W|(\frp)\;\chi[\text{(iii')}]\, \e^{a(\frp)}\leq \sum_{\frx_0}\int_{\frP} \mathrm{d}|W|(\frp)\;\chi\bigl[\exists\,\widetilde{\frp}\in\frP\bigl(\skl(\frp)\bigr);\,\widetilde{\frp}\leftrightarrow\frx_0\bigr]
\end{equation}
For every skeleton $\frs=\{\frx_1,\dots,\frx_{N}\}$ such that there is $\widetilde{\frp}\in\frP(\frs)$ with $\widetilde{\frp}\leftrightarrow\frx_0$ one finds
\begin{equation}
 \max_{T\in \caT^*_{N}}\prod_{\{i,j\}\in E(T)}\xi_{\alpha_1}(\frx_i,\frx_j)\leq\max_{T\in \caT^{}_{N}}\prod_{\{i,j\}\in E(T)}\xi_{\alpha_1}(\frx_i,\frx_j)
\end{equation}
By transferring the integral over polymers to an integral over clusters of constituents (just as it was done for case (i) and (ii') to obtain the first inequality in \eqref{eq: CaseTwo} and \eqref{eq: CaseOne} respectively) \eqref{eq: SplittingSingletons} is bounded by
\begin{equation}
 \begin{split}
  &\sum_{\frx_0}\sum_{N=1}^\infty { \frac{1}{N!}} \int_\frK \mathrm{d}w(\frx_1)\dots\int_\frK \mathrm{d}w(\frx_N) \max_{T\in \caT^{}_{N}}\prod_{\{i,j\}\in E(T)}\xi_{\alpha_1}(\frx_i,\frx_j)\\
  &\leq C''' \,L_{\mathrm{v}}(\frp')\,\de_2
 \end{split}
\end{equation}
for a constant $C'''$ that only depends on $\al_1,\al_2$. This finishes the proof, since the parameter $\de_2$ (same as in Lemma \ref{lem: EstConst}) can be chosen arbitrarily small for $\beta,\kappa$ large enough.\\
\end{proof}

\subsection{Construction of the classical potential}

The Koteck\'y--Preis criterion of Proposition \ref{prop: KoteckyPreis} allows to write the classical restriction $\mu_\La^{\be,X}$ in the form of \eqref{eq: polymerexpansion}, the polymer expansion. On the level of this polymer model, we moreover verified the conditions to proceed with a cluster expansion in the sense of \cite{Ue}, from where we extract what is relevant in our context in the following proposition. \\
We continue to suppress the dependence on $\La$, $\be$, $X$, and $x$ in the notation.

\begin{proposition}\label{prop: ClusterExpansion}
 For any choice of constants $0<c_1<g$ and $0<C_1,c_2,C_2$, there are $\kappa_{\min},\beta_{\min}>0$, so that, for any volume $\Lambda\Subset\bbZ^d$, any classical configuration $x\in \Omega$, and as long as $\beta\geq\beta_{\min}$, $\kappa\geq\kappa_{\min}$,
 \begin{equation} \label{eq: ClusterExpansion}
 \begin{split}
  \text{(1)}\quad &1+\sum_{N= 1}^\infty\frac{1}{N!}\int_{\frP} \mathrm{d}W(\frp_1)\dots \int_{\frP} \mathrm{d}W(\frp_N)\prod_{1\leq i<j\leq N}\chi[\frp_i \nleftrightarrow \frp_j]\\
  &=\exp \biggl[ \sum_{N=1}^\infty \frac{1}{N!}\int_{\frP} \mathrm{d}W(\frp_1)\dots \int_{\frP} \mathrm{d}W(\frp_N)\, \varphi(\frp_1,\dots,\frp_N)\biggr]\\
 \end{split}
 \end{equation}
 with
 \begin{equation}
  \varphi(\frp_1,\dots,\frp_N):=\left\{ \begin{array}{ll}
                                     1 &\text{if }N=1\\
                                     \sum_{G\in\caC_N}\prod_{(i,j)\in G}\bigl( - \chi[\frp_i\leftrightarrow \frp_j]\bigr)&\text{if }N\geq 2
                                    \end{array}\right. 
 \end{equation}
 where combined sum and integrals, the `integral over clusters', converge absolutely, and where $\caC_N$ denotes the set of connected graphs on the vertices $\{1,\dots,N\}$.
 
 (2) The `weight' of clusters decays exponentially in their length, i.~e., the integral of clusters adjacent to a polymer $\frp_0\in\frP$ can be bounded according to
 \begin{equation}\label{eq: ClusterBound}
  \begin{split}
   &\sum_{N=1}^\infty \frac{1}{N!}\int_{\frP} \mathrm{d}\bigl\lvert W\bigr\rvert (\frp_1)\dots \int_{\frP} \mathrm{d}\bigl\lvert W\bigr\rvert(\frp_N)\,\chi[\exists \, i \text{ with } \frp_0\leftrightarrow \frp_i]\\[-.2\baselineskip]
   &\hspace{.3\textwidth}\times |\varphi(\frp_1,\dots,\frp_N)|\, \prod_{i=1}^N \exp \bigl(c_1 L_\text{h}(\frp_i) + c_2 L_\text{v}(\frp_i)\bigr)\\
   &\leq C_1 L_\text{h}(\frp_0) + C_2 L_\text{v}(\frp_0)
  \end{split}
 \end{equation}
\end{proposition}

\begin{proof}
 This proposition is a consequence of Theorems 1 and 3 and equation (19) in \cite{Ue}, where the function $\zeta$ of this reference is given through $1+\zeta(\,\cdot\,,\,\cdot\,)=\chi[\,\cdot\nleftrightarrow \cdot\,]$. The conditions for these results to work are contained in Proposition \ref{prop: KoteckyPreis}.
 
\end{proof}

 Motivated by this result we abbreviate the integral over clusters of polymers, in notation $\frc=(\frp_1,\dots,\frp_N)\in\frC$, $N\geq 1$, as
\begin{equation}
 \int_{\frC} \mathrm{d}\caM(\frc)\;\cdot\;:= \sum_{N=1}^\infty\frac{1}{N!}\int_{\frP} \mathrm{d}W(\frp_1)\dots \int_{\frP} \mathrm{d}W(\frp_N)\, \varphi(\frp_1,\dots,\frp_N)\;\cdot\;
\end{equation}
and indeed $\caM$ is a consistently defined measure on $\frC(=\frC_\La^\be)$ for different $\La\Subset\bbZ^d$.  We also write $\caR(\frc)=\bigcup_i\caR(\frp_i)$ for the root-set of a cluster, $\mathrm{Dom}(\frc)=\bigcup_i\mathrm{Dom}(\frp_i)$ for its domain, $L_{\mathrm{v}/\mathrm{h}}(\frc)=\sum_i L_{\mathrm{v}/\mathrm{h}}(\frp_i)$ for its length, $\mathrm{span}_\mathrm{h}(\frc)$ for its horizontal span, i.~e.~the added minimal length of two intervals $I_\mathrm{l},I_\mathrm{r}\subset[0,\beta]$, so that $\mathrm{Dom}(\frc)\subset\Lambda\times I_\mathrm{l}\cup I_\mathrm{r}$, and $\frc\leftrightarrow \frp$ if it is adjacent to a polymer $\frp\in\frP$, i.~e., there is $i\in\{1,\dots,N\}$ with $\frp_i\leftrightarrow\frp$.\\

 We define the classical potential $\Psi$ as limit of the following finite volume approximations, depending on $\Lambda\Subset\bbZ^d$:
\begin{equation}\label{eq: ClassicalPotential}
 \Psi_{\Lambda,A}^{}(x):=\left\{\begin{array}{ll}
              \int_{\frC_\La} \mathrm{d}\caM_{}^{}(\frc)\,\chi[\cup_{i} \caR(\frp_i)=A]& \text{ if }|A|>1\\[2mm]
              \text{`as above' }+\log \Tr_i(\caP_i Q_i(x_i)) & \text{ if } A=\{i\}, i\in\Lambda\\[2mm]
              0 & \text{ if } A=\emptyset 
             \end{array}
  \right. 
\end{equation}
One way to see that the ${ \Psi_{\Lambda,A}^{}}$ are real is by the expansion's reflection symmetry with respect to the equal $\beta/2$-plane and, as desired, they only depend on $x_A\in\Omega_A$. 
\begin{theorem}\label{thm: GibbsLow2}
 Provided that the assumptions of Theorem~\ref{thm: gibbs1} hold and given a constant $c>0$, there exist $\kappa_{\min},\beta_{\min}>0$, such that, for any $\kappa\geq\kappa_{\min}$, $\beta\geq\beta_{\min}$, and $\La\Subset\bbZ^d$, the classical restriction takes the form
 \begin{equation}\label{eq: GibbsEnsemble}
  \begin{split}
   &\mu_\Lambda^{}(x_\Lambda)=\frac{1}{\tilde{Z}_\Lambda^{}}\exp\Bigl( \sum_{A\subset\Lambda} \Psi_{\Lambda,A}^{}(x_A) \Bigr)\\
  \end{split}
 \end{equation}
The (unique) thermodynamic limit $\mu$ of these Gibbs distributions is a Gibbs distribution for a potential given through
\begin{equation}\label{eq: ThDynLimPot}
 \Psi_{A}^{}(x_A):= \lim_{\Lambda\nearrow \bbZ^d} \Psi_{\Lambda,A}^{}(x_A)
\end{equation}
 which decays exponentially according to
 \begin{equation}\label{eq: ExpDecLim}
 \bigl\lVert \Psi \bigr\rVert_c<\infty  \quad \text{and} \quad  \sum_{A\ni 0} \max_{x_A\in\Om_A} \e^{c\cdot \mathrm{diam}(A)} \bigl\lvert\Psi_A (x_A)\bigr\rvert<\infty
 \end{equation}
 Moreover these statements remain true for the classical restriction of the ground state, i.e., for the (unique) probability distribution obtained by first taking $\beta\rightarrow \infty$ in~\eqref{eq: GibbsEnsemble} or after the thermodynamic limit in $\mu^{}$. The corresponding classical potential is given by \eqref{eq: ThDynLimPot} as $\be\rightarrow \infty$.
\end{theorem}

\begin{proof}
 \eqref{eq: GibbsEnsemble} follows if we summarize \eqref{eq: NonNormalClassProj}, \eqref{eq: polymerexpansion}, and \eqref{eq: ClusterExpansion} by
 \begin{equation}
  \Tr\bigl( Q(x) \e^{-\be H}\bigr)=\Tr\bigl(Q(x)\caP(\emptyset)\bigr)\exp\Bigl( \int_{\frC} \d \caM(\frc) \Bigr)
 \end{equation}
 and furthermore recall that the weight of a cluster $\frc=(\frp_1,\dots,\frp_n)$ with empty root set $\bigcup_i \frR(\frp_i)=\emptyset$ does not depend on the configuration $x\in\Omega$, so that the contribution of these \textit{bulk} clusters is canceled by normalization.\\
 For $\beta_{\min}$ large enough, we introduce two positive constants $c,C>0$, satisfying
\begin{equation}
 C:=C_1=C_2\quad \text{and }\quad c:=c_2\leq\frac{c_1}{2r}\beta_{\min}
\end{equation}
 in terms of the constants $c_1,c_2>0$, $c_1<g$, appearing in Proposition~\ref{prop: ClusterExpansion}. Denote with $\mathrm{span}_\mathrm{h}(\frp)$ the horizontal span of a polymer $\frp\in\frP$, i.~e.~the minimal added length of two intervals $I_\mathrm{l},I_\mathrm{r}\subset[0,\beta]$, so that $\mathrm{Dom}(\frp)\subset\Lambda\times I_\mathrm{l}\cup I_\mathrm{r}$. Using the bound~\eqref{eq: ClusterBound} we can estimate the difference of the classical potential for possibly different volumes $\Lambda'\subset\Lambda$ and temperatures $\beta'\leq\beta$ evaluated at the same $x\in\Omega$, $A\subset\Lambda'$:
\begin{equation}\label{eq: CauchyPotential}
\begin{split}
 &m_c(\Lambda',\Lambda,\beta',\beta)\,\bigl\lvert \Psi_{\Lambda',A}^{\beta'}(x) -\Psi_{\Lambda,A}^{\beta}(x) \bigr\rvert\\
 &\leq \sum_{a\in A}\int_{\frC} \mathrm{d}\left\lvert\caM\right\rvert(\frc)\, \chi[\textstyle{\mathrm{Dom}_r(\frc)\cap\bigl((\Lambda\setminus\Lambda')\times[0,\beta]\bigr)\neq\emptyset \text{ or }\mathrm{span}_\mathrm{h}(\frc)> \beta'}]\\
 &\hphantom{\leq \sum_{a\in A}\int_\frC }\chi[\textstyle{a\in \caR(\frc)}]\, \exp \bigl(c_1 L_\text{h}(\frc) + c_2 L_\text{v}(\frc)\bigr)\\
 &\hphantom{\leq}+\sum_{a\in A}\int_{\frC'} \mathrm{d}\left\lvert\caM \right\rvert(\frc)\, \chi[\mathrm{span}_\mathrm{h}(\frc)= \beta']\,\chi[\textstyle{a\in \caR(\frc)}]\,\exp \bigl(c_1 L_\text{h}(\frc)\bigr)\\
 &\leq 2\,C\,|A|
 \end{split}
 \end{equation}
 with $\frC=\frC_\La^\be$, $\frC'=\frC_{\La'}^{\be'}$, and with
 \begin{equation}
  m_c(\Lambda',\Lambda,\beta',\beta):=\max\bigl\{ \chi[\Lambda'\neq\Lambda]\,\exp(c\,\mathrm{dist}(\Lambda\setminus\Lambda',A))\,,\,\chi[\beta'\neq\beta]\,\exp(c\,\beta') \bigr\}\\
 \end{equation}
 In terms of the graphical representation, note that the above difference is merely an integral over those clusters $\frc\in\frC$, rooted in $A$ which are end-to-end clusters (second term), which have a horizontal span greater than $\beta'$, or which reach vertically into the complemental volume $\Lambda\setminus\Lambda'$ through vertical segments or through horizontal segments with an effective vertical range $r$. Recall again our assumption of all interaction sets $B_i$ being connected. By the choice of the constants the contributing clusters satisfy either
 \begin{equation}
 \begin{split}
  &\beta' \leq L_\text{h}(\frc) \quad\text{ and/or}\\
  &c\cdot\mathrm{dist}(\Lambda\setminus\La',A)\leq c\, L_\text{v}(\frc)+2r c\,|I(\frp_i)| \leq  c_1\, L_\text{h}(\frc)+c_2\, L_\text{v}(\frc)
 \end{split}
 \end{equation}
 which shows how we could absorb the factor $m_c$ in the bounding integral. To obtain the second inequality in \eqref{eq: CauchyPotential} we covered the root-set $A$ with $|A|$ polymer `singletons' $(t_1,\emptyset,\{i\})\in\frP_1$, which play the role of the fixed polymer $\frp_0$ in~\eqref{eq: ClusterBound}.\\
 We have thus proven the existence the thermodynamic limit \eqref{eq: ThDynLimPot} for each $A\Subset\bbZ^d$, which can be understood as integral over clusters (of course with finite length) in $\bbZ^d\times [0,\beta]$ rooted in $A$, and furthermore that it is interchangeable with the limit $\beta\rightarrow \infty$. In this limit the contribution of end-to-end clusters vanishes exponentially and by the cyclicity of the trace we may think of the classical potential at zero temperature $\Psi_A^{\infty}$ as integral over clusters in $\bbZ^d\times\bbR$ that have contact with the $\beta=0$ plane at positions in $A$. In the following we always allow $\beta=\infty$.\\
 The exponential decay property~\eqref{eq: ExpDecLim} can be read as integral over all clusters $\frc\in\frC$, that are rooted in $A\ni 0$ and respectively weighted with the exponential of
 \begin{equation}
  c\cdot\lvert A\rvert, \,c\cdot \mathrm{diam} (A)\leq c_1 \,L_\text{h}(\frc)+c_2\, L_\text{v}(\frc)
 \end{equation}
 and this integral can be bounded from above by the constant $C$ by using again the estimate in Proposition~\ref{prop: ClusterExpansion} similarly as in~\eqref{eq: CauchyPotential}.\\
We can now immediately conclude that, for $\Gamma\subset \Lambda\Subset\bbZ^d$, the conditional probabilities,
\begin{equation}
 \mu_\Lambda^{}\bigl(x_\Gamma\,|\,x_{\Lambda\setminus\Gamma}\bigr):=\text{(norm.)}\times\exp\biggl( \sum_{\substack{A\subset\Lambda\\A\cap\Gamma\neq\emptyset}}\Psi_{\Lambda,A}^{}(x_A) \biggr)
\end{equation}
converge uniformly in $x\in \Omega$ as $\Lambda\nearrow \bbZ^d$. Almost by definition, this proves that any thermodynamic limit point of $\mu_\Lambda^{}$ is a Gibbs distribution for the potential $\Psi^{}$, see e.~g.~\cite{Ru,Si} for standard arguments.\\
 Expectation values with respect to $\mu_\Lambda^{}$ of local functions, say only depending on $x_\Gamma\in\Omega_\Gamma$, $\Gamma\subset\Lambda$, converge as $\Lambda\nearrow \bbZ^d$ (again interchangeable with $\beta\rightarrow \infty$), which can be verified by beginning right from the start to work with $ Q_\Gamma(x_\Gamma)\otimes \opunit_{\Lambda\setminus\Gamma}$ instead of $Q_\Lambda(x_\Lambda)$ in the definition of the classical restriction. With this replacement, which does not harm the previous constructions, one obtains the marginal distribution of the classical restriction as
 \begin{equation}
 \begin{split}
  &\mu_\La(x_\Ga)=\exp \biggl[ \int_\frC \mathrm{d}\caM(\frc)\,\chi[\textstyle{\caR(\frc)\cap \Gamma\neq\emptyset}]\biggr]
  \end{split}
 \end{equation}
 where we have abused the notation, as the measure on the {\small RHS} is now defined with respect to the `inhomogeneous observable' with $X_i=X$ for $i\in\Gamma$, and $X_i=\opunit$ at sites from the complement $\Ga^\complement$, and for classical configurations of the form $x=x_\Ga\equiv(x^{}_\Gamma,\boldsymbol{1}_{\Ga^\complement})$. The contribution from clusters which are not rooted in $\Gamma$ is again canceled by normalization. By the same arguments as earlier in this proof, mainly the exponential decay of the cluster weights, the above expression has a well-defined thermodynamic limit, which is interchangeable with taking $\beta\rightarrow\infty$.\\
\end{proof}

\section{The ground state of the Ising chain in a transverse field}

In this section we prove  Theorems \ref{thm: nongibbs} and \ref{thm: LdpForIsing} concerning a non-locality property of the ground state of the Ising chain in a transverse field. The origin of this non-locality is easily  understood in finite volume, as we explain now:\\
Note that the `parity operator' $P:=\exp(i\pi \sum_l \sigma_l^z)$ commutes with the local Hamiltonian $H_\La$. In volumes $\La$ consisting of an even number of sites and for $J=0$ the non-degenerate ground state $\omega_\La(\,\cdot\,)=\langle \psi_\mathrm{gs}|\,\cdot\,|\psi_\mathrm{gs}\rangle$ has positive parity in the sense that it is an eigenstate of $P$ for the eigenvalue $p=+1$. By simple perturbation theory the gapped ground state maintains positive parity for $|J/h|<1$. For $X=\sigma^z$, the classical restriction $\mu_\Lambda^{X}(x)$, $x\in\Omega_\La=\{-1,+1\}^\Lambda$, then vanishes whenever the number of spins facing the same direction or equivalently whenever $\sum_l x_l$ is odd. This is clearly a non-local effect and the core of our argument is to show that this nonlocality persists in infinite volume. \\
In the following we always have in mind the choice $X=\sigma^z$ and as in the introduction we write $\mu^z$ for the belonging classical classical restriction of the ground state.

\subsection{Absence of quasi-locality}
Instead of fermionizing the spin in a Jordan-Wigner-tranformation as is commonly done for solving this model explicitly, see the Appendix, we use  the previously presented cluster expansion which is not restricted to spin chains. We treat here the Ising model for a slightly modified Hamiltonian,
\begin{equation}
 H_\Lambda=\sum_{i}(\sigma_i^z+\opunit)-\e^{-2\kappa}\sum_{i} \sigma_i^x\sigma_{i+1}^x
\end{equation}
to make the quantum interaction exactly in line with Assumption 1 of Theorem \ref{thm: gibbs1}. For notational purposes we only treat the one-dimensional setting $i\in\bbZ$ explicitly, but it is straightforward to check that the proof given here carries over to higher dimensions. 

Once again, note that the (infinite volume) ground state $\om$ of the Ising chain in transverse field is unique; see~\cite{AM}, and that its classical restriction equals the limit 
\begin{equation}
  \mu^z=\lim_{\Lambda\nearrow\bbZ,\, \beta \rightarrow \infty}\mu_{\Lambda}^{\beta ,X},\quad X=\sigma^z
 \end{equation}in arbitrary order. As before we mostly keep the dependence on $\be$ and $X=\sigma^z$ implicit in notation and write $\mu_\La=\mu_{\La}^{\be,X}$. Recall the notation introduced in Section \ref{subsec: Gibbsianess} and in particular the notion of absence of quasi-locality as in \eqref{eq: locality concrete}. For $L> 1$  we set $\Ga_L:=\{-L^2,\dots,L^2\}\subset \bbZ$, but we often suppress the subscript $L$ as in the following proposition. Cylinder sets of configurations on the infinite lattice are abbreviated by their defining constraint.
 
 \begin{proposition}\label{prop: CondProp}
  Given $\kappa>0$ large enough, the conditional probabilities of the classical restriction $\mu^z$ satisfy
  \begin{equation}\label{eq: CondPropZero}
   \mu^z\bigl( \{x_0=+1\} \,|\, \{x_{\Ga\setminus\{0\}}\equiv -1\}\bigr)\; \xrightarrow[]{\; L\rightarrow \infty\;} 0
  \end{equation}
  and
  \begin{equation}\label{eq: CondPropOne}
   \mu^z\bigl( \{x_0=+1\} \,|\, \{x_L=+1,\,x_{\Ga\setminus\{0,L\}}\equiv -1\}\bigr)\; \xrightarrow[]{\; L\rightarrow \infty\;} 1
  \end{equation}
  and the above expressions are well-defined, since $\mu^z$ is positive on each cylinder set. Therefore $x\in\Om$ defined through $x_i=-1$, $i \in\bbZ$, is a bad configuration in the sense of section \ref{subsec: Gibbsianess}.
 \end{proposition}
 
By proceeding just as before in the general setting we can again express the classical restriction $\mu^{}_\La$ in terms of a polymer model. As a pecularity of the Ising model the polymers can be seen to be non-intersecting loops (in particular without ends) which furthermore have non-negative polymer weights, \ie, the density $\rho$ as defined in \eqref{eq: density} and \eqref{eq: density0} is non-negative. We want to give a rough sketch of the proof for the above result in terms of the diagrammatic language of such a loop gas. The exact details will be supplied only in the next section.

 \begin{figure}[h] 
 \subfigure[]{
  \def\svgwidth{.3\textwidth}
\begingroup%
  \makeatletter%
  \providecommand\color[2][]{%
    \errmessage{(Inkscape) Color is used for the text in Inkscape, but the package 'color.sty' is not loaded}%
    \renewcommand\color[2][]{}%
  }%
  \providecommand\transparent[1]{%
    \errmessage{(Inkscape) Transparency is used (non-zero) for the text in Inkscape, but the package 'transparent.sty' is not loaded}%
    \renewcommand\transparent[1]{}%
  }%
  \providecommand\rotatebox[2]{#2}%
  \ifx\svgwidth\undefined%
    \setlength{\unitlength}{176.1bp}%
    \ifx\svgscale\undefined%
      \relax%
    \else%
      \setlength{\unitlength}{\unitlength * \real{\svgscale}}%
    \fi%
  \else%
    \setlength{\unitlength}{\svgwidth}%
  \fi%
  \global\let\svgwidth\undefined%
  \global\let\svgscale\undefined%
  \makeatother%
  \begin{picture}(1,1.72671777)%
    \put(0,0){\includegraphics[width=\unitlength]{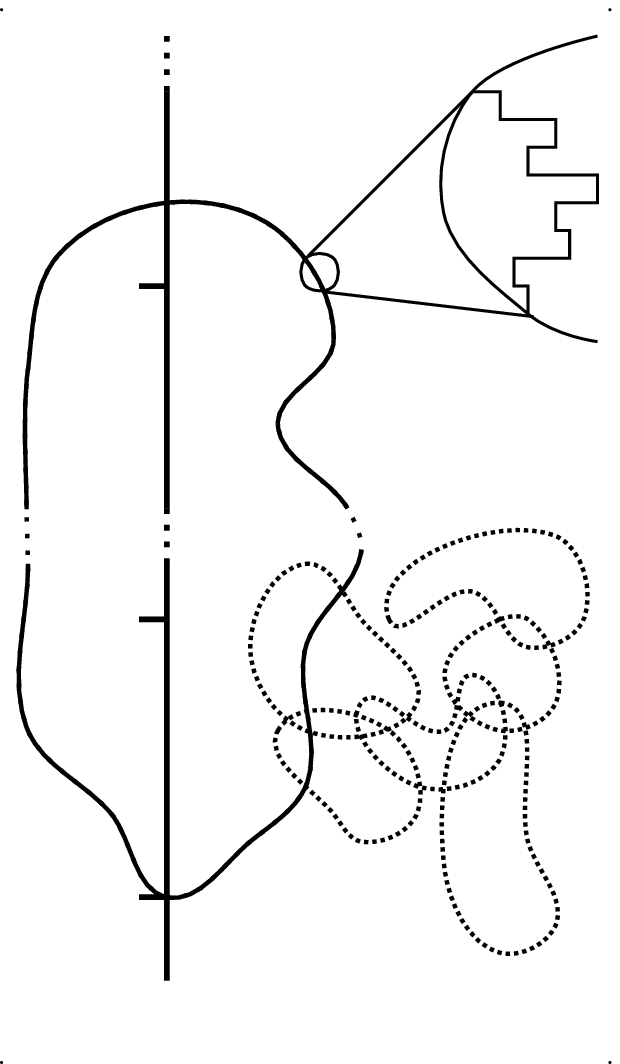}}%
    \put(0.22418265,0.04562893){\color[rgb]{0,0,0}\makebox(0,0)[lb]{\smash{$\beta=0$}}}%
    \put(0.13657013,0.27277259){\color[rgb]{0,0,0}\makebox(0,0)[lb]{\smash{$0$}}}%
    \put(0.13657013,0.72705993){\color[rgb]{0,0,0}\makebox(0,0)[lb]{\smash{$L$}}}%
    \put(0.13657013,1.27220474){\color[rgb]{0,0,0}\makebox(0,0)[lb]{\smash{$L^2$}}}%
    \put(0.13657013,1.54477714){\color[rgb]{0,0,0}\makebox(0,0)[lb]{\smash{$\bbZ$}}}%
    \put(0.55740882,1.09069845){\color[rgb]{0,0,0}\makebox(0,0)[lb]{\smash{$\frp_+$}}}%
    \put(0.72714367,0.90877487){\color[rgb]{0,0,0}\makebox(0,0)[lb]{\smash{cluster $\frc$}}}%
  \end{picture}%
\endgroup%
}
 \subfigure[]{
  \def\svgwidth{.3\textwidth}
\begingroup%
  \makeatletter%
  \providecommand\color[2][]{%
    \errmessage{(Inkscape) Color is used for the text in Inkscape, but the package 'color.sty' is not loaded}%
    \renewcommand\color[2][]{}%
  }%
  \providecommand\transparent[1]{%
    \errmessage{(Inkscape) Transparency is used (non-zero) for the text in Inkscape, but the package 'transparent.sty' is not loaded}%
    \renewcommand\transparent[1]{}%
  }%
  \providecommand\rotatebox[2]{#2}%
  \ifx\svgwidth\undefined%
    \setlength{\unitlength}{176.1bp}%
    \ifx\svgscale\undefined%
      \relax%
    \else%
      \setlength{\unitlength}{\unitlength * \real{\svgscale}}%
    \fi%
  \else%
    \setlength{\unitlength}{\svgwidth}%
  \fi%
  \global\let\svgwidth\undefined%
  \global\let\svgscale\undefined%
  \makeatother%
  \begin{picture}(1,1.72671777)%
    \put(0,0){\includegraphics[width=\unitlength]{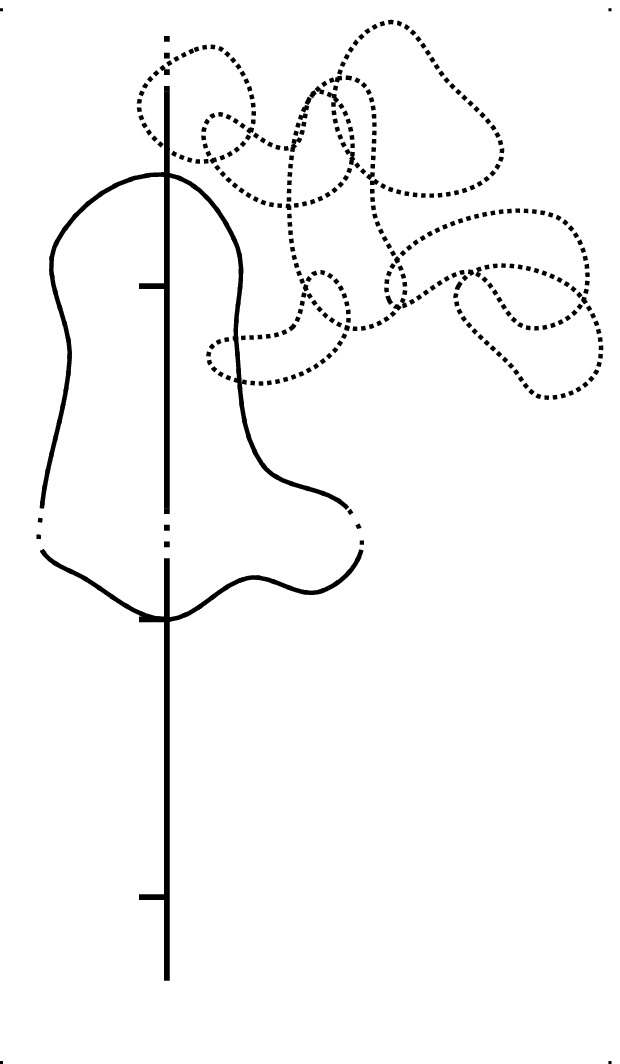}}%
    \put(0.22418265,0.04562893){\color[rgb]{0,0,0}\makebox(0,0)[lb]{\smash{$\beta=0$}}}%
    \put(0.13657013,0.27277259){\color[rgb]{0,0,0}\makebox(0,0)[lb]{\smash{$0$}}}%
    \put(0.13657013,0.72705993){\color[rgb]{0,0,0}\makebox(0,0)[lb]{\smash{$L$}}}%
    \put(0.13657013,1.27220474){\color[rgb]{0,0,0}\makebox(0,0)[lb]{\smash{$L^2$}}}%
    \put(0.13657013,1.54477714){\color[rgb]{0,0,0}\makebox(0,0)[lb]{\smash{$\bbZ$}}}%
    \put(0.4318569,1.0223467){\color[rgb]{0,0,0}\makebox(0,0)[lb]{\smash{$\frp_+$}}}%
    \put(0.72714367,0.90877487){\color[rgb]{0,0,0}\makebox(0,0)[lb]{\smash{cluster $\frc$}}}%
  \end{picture}%
\endgroup%
  }
 \subfigure[]{
  \def\svgwidth{.3\textwidth}
\begingroup%
  \makeatletter%
  \providecommand\color[2][]{%
    \errmessage{(Inkscape) Color is used for the text in Inkscape, but the package 'color.sty' is not loaded}%
    \renewcommand\color[2][]{}%
  }%
  \providecommand\transparent[1]{%
    \errmessage{(Inkscape) Transparency is used (non-zero) for the text in Inkscape, but the package 'transparent.sty' is not loaded}%
    \renewcommand\transparent[1]{}%
  }%
  \providecommand\rotatebox[2]{#2}%
  \ifx\svgwidth\undefined%
    \setlength{\unitlength}{176.1bp}%
    \ifx\svgscale\undefined%
      \relax%
    \else%
      \setlength{\unitlength}{\unitlength * \real{\svgscale}}%
    \fi%
  \else%
    \setlength{\unitlength}{\svgwidth}%
  \fi%
  \global\let\svgwidth\undefined%
  \global\let\svgscale\undefined%
  \makeatother%
  \begin{picture}(1,1.72671777)%
    \put(0,0){\includegraphics[width=\unitlength]{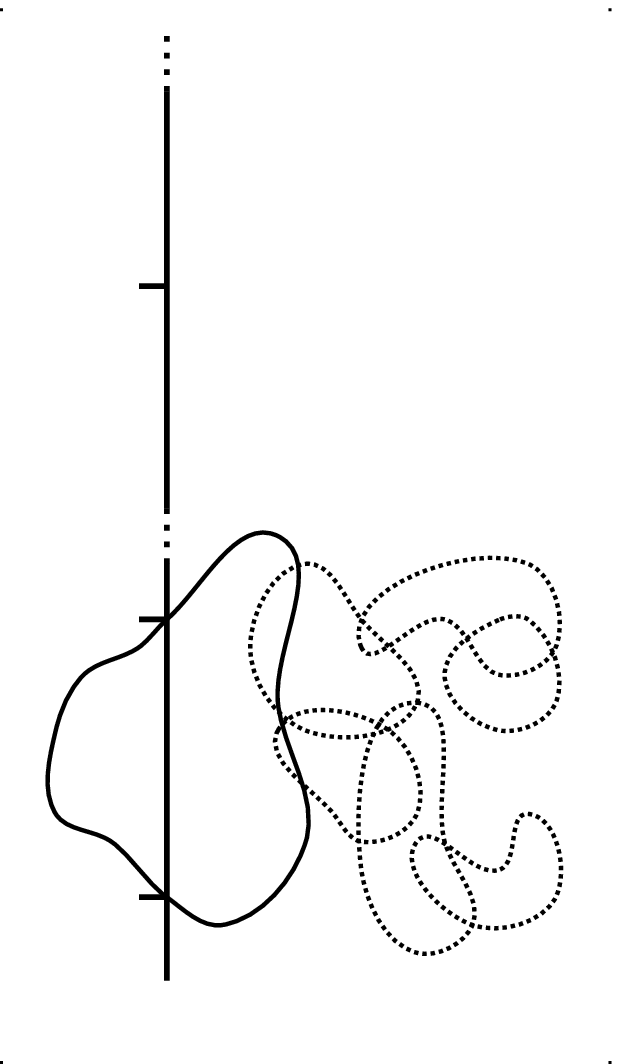}}%
    \put(0.22418265,0.04562893){\color[rgb]{0,0,0}\makebox(0,0)[lb]{\smash{$\beta=0$}}}%
    \put(0.13657013,0.27277259){\color[rgb]{0,0,0}\makebox(0,0)[lb]{\smash{$0$}}}%
    \put(0.13657013,0.72705993){\color[rgb]{0,0,0}\makebox(0,0)[lb]{\smash{$L$}}}%
    \put(0.13657013,1.27220474){\color[rgb]{0,0,0}\makebox(0,0)[lb]{\smash{$L^2$}}}%
    \put(0.13657013,1.54477714){\color[rgb]{0,0,0}\makebox(0,0)[lb]{\smash{$\bbZ$}}}%
    \put(0.38642817,0.8860605){\color[rgb]{0,0,0}\makebox(0,0)[lb]{\smash{$\frp_+$}}}%
    \put(0.72714367,0.90877487){\color[rgb]{0,0,0}\makebox(0,0)[lb]{\smash{cluster $\frc$}}}%
  \end{picture}%
\endgroup%
  }
  \caption{Graphical representation of contributions to \\
  (a) $\mu^z \bigl( \{x_0=+1\}\,\bigl\lvert\,\{x_{\Gamma\setminus\{0\}}\equiv-1\} \bigr)$,\\ 
  (b) $\mu^z \bigl( \{x_0=-1,\,x_L=+1\}\,\bigl\lvert\,\{x_{\Gamma\setminus\{0,L\}}\equiv-1\} \bigr)$,\\
  (c) $\mu^z \bigl( \{x_0=x_L=+1\}\,\bigl\lvert\,\{x_{\Gamma\setminus\{0,L\}}\equiv-1\} \bigr)$.}
  
  \label{fig: LoopGas}
\end{figure}

As we will see, the conditional probability~\eqref{eq: CondPropZero} can be read as integral over loops $\frp_+$ in $\bbZ\times \bbR$ which are `dressed' with clusters of other loops and which are pinned to the origin $(0,0)$ but forbidden to touch $(\Gamma\setminus \{0\})\times 0$, corresponding to the condition $x_i=-1$, for $0<|i|<L$, see (a) in Fig~\ref{fig: LoopGas}. We will also see that the loop $\frp_+$ must reach spatially from the origin into the complement of $\Gamma$, and by the loop weight's exponential decay in its length we get that~\eqref{eq: CondPropZero} decays exponentially in $L^2$. \\
The limiting behaviour~\eqref{eq: CondPropOne} is equivalent to a vanishing ratio
\begin{equation}\label{eq: Ratio}
 \frac{\mu^z \bigl( \{x_0=-1,\,x_L=+1\}\,\bigl\lvert\,\{x_{\Gamma\setminus\{0,L\}}\equiv-1\} \bigr)}{\mu^z \bigl( \{x_0=x_L=+1\}\,\bigl\lvert\,\{x_{\Gamma\setminus\{0,L\}}\equiv-1\} \bigr)}\; \xrightarrow[]{\; L\rightarrow \infty\;} 0
\end{equation}
The numerator is illustrated in (b) of Fig~\ref{fig: LoopGas} and goes to zero exponentially in $L^2-L$, which is the distance between site $L$ and the complement of $\Gamma$. In the diagrammatic representation of the denominator in Fig~\ref{fig: LoopGas}~(c) the contributing loops must cross the $\beta=0$ plane at $0$ and $L$, and therefore it cannot decay faster than exponentially in $L$.

\subsubsection{Proof of Proposition~\ref{prop: CondProp}} $\hphantom{m}$\\

1) The classical restriction as loop-gas:\\
Assume throughout that the values of $\beta$ and $\kappa$ are large enough in the sense of Theorem~\ref{thm: GibbsLow2}. Then we apply the results of the previous part of this low temperature section for the trivial single-site observable $X=\opunit$ (giving a configuration space of only one element $\Om=\{x\equiv \funit\}$, see also the last paragraph in the proof of Theorem \ref{thm: GibbsLow2}), in particular Proposition~\ref{prop: ClusterExpansion}, to express the partition function in the form
\begin{equation}\label{eq: PartitionFunc}
\begin{split}
  &Z_\Lambda^{}=\mathrm{Tr}_\Lambda\left(\e^{ -\beta H_\Lambda} \right)\\
  &=1+\sum_{N=1}^\infty\frac{1}{N!}\int_{\frP} \mathrm{d}\caL(\frp_1)\dots \int_{\frP} \mathrm{d}\caL(\frp_N)\;\chi[\frp_i \nleftrightarrow \frp_j,\;1\leq i<j\leq N]\\
 & =\exp \biggl[\int_{\frC} \mathrm{d}\caM(\frc)\biggr]
\end{split}
\end{equation}
where we introduced $\caL:=W$ to emphasize the following peculiarity: since $\sigma^x$ flips the spin in $z$-basis, i.~e., we have $\sigma^x=\caP\sigma^x(\opunit-\caP) +(\opunit-\caP)\sigma^x\caP$, the measure $\caL$ is non-vanishing only on the set of polymers that are diagrammatically represented by (closed) Loops on $\Lambda\times[0,\beta]$ if we identify the points $(i,0)\sim(i,\beta),i\in\Lambda$. Recall that $\caL$ and $\caM$ implicitly depend on the inverse temperature $\be$ and classical configuration which here is always taken to be $x\equiv \funit$. Furthermore $\caL$ is translation invariant on these cylinders and, if restricted to `contractable'  loops, mutually consistent, i.~e.~the (positive) density of a loop with certain `shape' does not depend on the ambient cylinder.\\
The spectral projection of $X=\sigma^z$ to $x=-1$ trivially equals the local ground state projection of the uncoupled Hamiltonian, i.~e., $Q(-1)=\caP$, and if we repeat the procedure of Section~\ref{sub: PolymerModel} we may use that for any diagram $\frX\in \frS_n$, $n\geq 0$, 
\begin{equation}
 \caR(\frX)\neq p_\Lambda(x):=\{i\in\Lambda\,|\,x_i=+1\}\;\text{ implies }\; \rho(\frX)=0
\end{equation} With this implicit one-to-one correspondence between configurations and root-sets we can express the (marginal) probabilities as restrictions of the partition function~\eqref{eq: PartitionFunc} by imposing conditions on the root-sets of the involved loops and clusters. Let $p,m\subset\Lambda$, $p\cap m=\emptyset$, denote sets of `plus-sites' and `minus-sites' respectively. Since $\caL$ is non-negative there are no convergence concerns when writing
\begin{equation}\label{eq: MargProp}
 \begin{split}
  &\mu_\Lambda^{}\bigl( \{x_p\equiv+1,\, x_m\equiv-1\} \bigr)\\
  &=\frac{1}{Z_\Lambda}\Bigl\{\chi[p=\emptyset] +\sum_{N=1}^\infty\frac{1}{N!}\int_{\frP} \mathrm{d}\caL(\frp_1)\dots \int_{\frP} \mathrm{d}\caL(\frp_N)\;\chi[\cup_i \caR(\frp_i)\cap m=\emptyset] \\
  &\hphantom{=\frac{1}{Z_\Lambda^\beta} \sum_{N=1}^\infty\frac{1}{N!}\int_{\frP} \mathrm{d}\caL_{\Lambda,\beta}(\frp_1)\dots \int_{\frP}}\chi[\frp_i \nleftrightarrow \frp_j,\;1\leq i<j\leq N]\Bigl\}\\
  &=\sum_{N=1}^{|p|}\frac{1}{N!} \int_{\frP} \mathrm{d}\caL(\frp_{1})\dots \int_{\frP} \mathrm{d}\caL(\frp_{N})\\
  &\hphantom{=\sum_{N=\lceil |p|/2\rceil}^{|p|}\frac{1}{N!} \int_{\frP}}\chi[\cup_i \caR(\frp_{i})\cap m=\emptyset]\;\chi[\frp_i \nleftrightarrow \frp_j,\;1\leq i<j\leq N]\\
  &\hphantom{=\sum_{N=\lceil |p|/2\rceil}^{|p|}\frac{1}{N!} \int_{\frP}}\chi[p\subset \cup_i \caR(\frp_i)]\;\chi[p\cap \caR(\frp_i)\neq\emptyset,\,i=1,\dots,N]\\
  &\hphantom{=\sum_{N=\lceil |p|/2\rceil}^{|p|}\frac{1}{N!} \int_{\frP}}\exp\biggl[- \int_{\frC} \mathrm{d}\caM(\frc)\,\chi[\exists i\; \text{with }\frp_i\leftrightarrow\frc \text{ or } \caR(\frc)\cap m\neq\emptyset]\biggr]
 \end{split}
\end{equation}
where in case $p=\emptyset$ we read the above sum as the plain exponential. The last expression is the announced integration over the particular loops which are rooted at `plus-sites' $p$ and dressed with clusters and requires further explanation: unless no configuration is fixed to be $+1$, i.~e., $p=\emptyset$, the polymer expansion of the second equality cannot be processed in our type of cluster expansion, since then the weight on the \textit{empty diagram} (without loops) vanishes.\\
If we fix for a moment the at most $|p|$ different loops which have roots in $p$, we are left with an integration over loops which must not be adjacent to these separated loops. There may also not be an additional loop besides the ones we fixed (a new empty diagram) and we can thus rewrite this remaining polymer expansion as before as exponential of an integral over clusters which must not be adjacent to the separated loops or rooted in $m$. These contributions are all canceled by the normalization factor $Z_\Lambda$ as in~\eqref{eq: PartitionFunc}, leaving behind clusters that are indeed adjacent to the separated loops or rooted in $m$. This also explains the minus sign in the above exponential.\\

2) The limits $\beta\rightarrow \infty$ and $\Lambda \nearrow \bbZ$:\\ By the exponential decay of the loops and clusters, the (infinite volume) ground state's classical restriction can be visualized by a gas of dressed loops on $\bbZ\times\bbR$: the loops must cross the $\beta=0$-line and the dressing clusters must be adjacent to either these secluded loops and/ or to the set $m\times \{0\}$.\\

3) Bounds on conditional probabilities:\\
We now express in formulae what was said about the conditional probabilities illustrated in Fig~\ref{fig: LoopGas}. Using~\eqref{eq: MargProp} we write (a) as
\begin{equation}
 \begin{split}
  &\mu^z \bigl( \{x_0=+1\}\,\bigl\lvert \{x_{\Gamma\setminus\{0\}}\equiv-1\} \bigr)\\
  &=\frac{\mu^z\bigl(\{x_0=+1,\,x_{\Gamma\setminus 0}\equiv -1\}\bigr)}{\mu^z\bigl(\{x_{\Gamma\setminus 0}\equiv -1\}\bigr)}\\
  &=\lim_{\Lambda\nearrow\bbZ,\,\beta\rightarrow\infty}\int_{\frP} \mathrm{d}\caL(\frp_+) \,\chi[ \caR(\frp_+)\cap\Gamma=\{0\}]\\
  &\hphantom{=\lim_{\Lambda\nearrow\bbZ,\,\beta\rightarrow\infty}\int_\frP}\exp\biggl[- \int_{\frC} \mathrm{d}\caM(\frc)\,\chi\left[\frp_+\leftrightarrow\frc\right] \chi\left[\caR(\frc)\cap (\Gamma\setminus\{0\})=\emptyset\right]\biggr]\\
  &=\caO\Bigl(\e^{-C \,L^2}\Bigr)
 \end{split}
\end{equation}
for a positive constant $C$. The denominator canceled the clusters which have roots in $m=\Gamma\setminus\{0\}$ so that the above conditional probability is indeed an integral over one loop $\frp_+$ crossing the $\beta=0$ line at the origin and outside of $\Gamma$ and dressed with clusters which must be adjacent to $\frp_+$ but not rooted in $m$. The exponential upper bound follows from Proposition~\ref{prop: ClusterExpansion} and direct application of Proposition~\ref{prop: KoteckyPreis} and analogously for (b), 
\begin{equation}
 \mu^z \bigl( \{x_0=-1,\,x_L=+1\}\,\bigl\lvert\,\{x_{\Gamma\setminus\{0,L\}}\equiv-1\} \bigr)=\caO\Bigl(\e^{-C \,(L^2-L)}\Bigr).
\end{equation}
To prove Proposition~\ref{prop: CondProp} requires a lower bound on the conditional probability in the denominator of~\eqref{eq: Ratio}. By proceeding as above and neglecting most (positive) contributions in the integral over dressed loops, we find, with suitable constants $c,c',c''>0$,
\begin{equation}
 \begin{split}
  &\mu^z \bigl( \{x_0=x_L=+1\}\,\bigl\lvert\,\{x_{\Gamma\setminus\{0,L\}}\equiv-1\} \bigr)\\
  &\geq \lim_{\Lambda\nearrow\bbZ,\,\beta\rightarrow\infty} \int_{\frP} \mathrm{d}\caL(\frp_+) \chi[\caR(\frp_+)=\{0,L\}]\,\chi[L_\mathrm{v}(\frp_+)=2L]\,\chi[\mathrm{span}_\mathrm{h}(\frp_+)\leq 2]\\
  &\hphantom{\geq \lim_{\Lambda\nearrow\bbZ,\,\beta\rightarrow\infty} \int \mathrm{d}} \exp\biggl[- \underbrace{\int_{\frC} \mathrm{d}|\caM|(\frc)\,\chi\left[\frp_+\leftrightarrow\frc\right]}_{\substack{\leq\; c'' (L_\mathrm{v}(\frp_+)+L_\mathrm{h}(\frp_+))}}\biggr]\\
  &\geq c'\cdot\e^{-c\,L}
 \end{split}
\end{equation}
To obtain the last inequality, we have used the definition of the loop weights and that the above restricted integral with respect to $\caL$ really means to integrate  a positive function on the $2L$-dimensional unit-cube, which is bounded below by $\exp\bigl(-4c'' L -2(\kappa+1)L\bigr)$.

\subsection{Large deviation principle for the magnetization: \\ Proof of Theorem \ref{thm: LdpForIsing}}
 
 The ground state of the Ising model can be determined explicitly using a Jordan--Wigner-transformation. This computation is summarized in Appendix~\ref{apx: GSising}, where, following~\cite{CEM,CD}, it is also shown that the expectation,
 \begin{equation}
  \omega\left( \exp\bigl( t{\textstyle \sum_{j=1}^n} \sigma_j^z \bigr) \right)=\det\left( M_n^t \right)
 \end{equation}
can be written as determinant of an $n\times n$ Toeplitz-matrix,
\begin{equation}
 \bigl(M_n^t\bigr)_{jj'}=\widehat{\phi_t}(j-j'):=\int_{-\pi}^\pi\,\frac{\mathrm{d}k}{2\pi}\,\phi_t(k)\e^{-ik(j-j')}
\end{equation}
for the \emph{symbol}
\begin{equation}
 \phi_t(k)= \cosh(t)-\sinh(t) \frac{h/J+\e^{-i k}}{\sqrt{(h/J+\e^{-ik})(h/J+\e^{ik})}}
\end{equation}
Note that, for any $t \in \bbR$, $\mathrm{Re}( \phi_t)>0$, and that it is analytic (take the positive branch of the square-root) as a function on a sufficiently thin ring domain containing the complex unit-circle $z=\e^{ik}$. Then also $\log \phi_t$ is analytic on such a domain and the Fourier-coefficients of $\log \phi_t$, which are nothing but the Laurent-coefficients, decay exponentially fast. In this case a strong type of Szeg\H{o}'s Theorem, see e.g.~\cite{IKD}, yields
\begin{equation}
 F(t)=\widehat{ \log \phi_t }(0)
\end{equation}
which is differentiable in $t\in\bbR$ by Leibniz's rule.


\begin{appendix}

\section{Solving the ground state of the Ising chain in a transverse field}\label{apx: GSising}
For real parameters $h,J$ satisfying $\lvert g \rvert> 1$ for $g:=h/J$,  the Ising model in a transverse field defined in Section 2, see \eqref{quantum ising}, has a unique (infinite-volume) ground state $\omega$, see \cite{AM,M}. This ground state can be obtained as the unique weak${}^*$ limit, as $\Lambda\nearrow \bbZ$, of the ground states $\omega_\Lambda$ in finite volumes on $\caA_\Lambda$ (extended by zero to a state on $\caA$). \\
Therefore we may take $\Lambda_N=\{-N/2,\dots,N/2-1\}$ with an even number of sites and the Hamiltonian may be modified at the boundaries of these chains without effect on the limit point. We  also use the abbreviations $\Lambda_N'=\{-N/2,\dots,N/2-2\}$ and $\omega_N=\omega_{\Lambda_N}$. Besides the important constraint $\lvert g \rvert> 1$ we furthermore set $J>0$ however merely for convenience.

\subsection{Jordan--Wigner-transformation} For now it is convenient to work with periodic boundary conditions by identifying $\sigma_{N/2}^x\equiv\sigma_{-N/2}^x $ in the Hamiltonian 
\begin{equation}
\begin{split}
 H_N=&-\sum_{j\in\Lambda_N} h \sigma_j^z + J \sigma_j^x \sigma_{j+1}^x
\end{split}
\end{equation}
As usual, we introduce the operators
\begin{equation}
 \begin{split}
  &a_j^*=\sigma_j^+ \exp\left(-i\pi \sideset{}{_{l<j}}\sum \sigma_l^+\sigma_l^- \right),\quad j\in\Lambda_N
 \end{split}
\end{equation}
which are defined in terms of the spin raising/lowering operators $\sigma_j^\pm=\frac{1}{2}(\sigma_j^x\pm \sigma_j^y)$ and which together with their adjoints satisfy the canonical anticommutation relations ({\small CAR}). With these fermion creation/annihilation operators the Hamiltonian can be rewritten as
\begin{equation}
\begin{split}
 H_N= -h\sum_{j\in\Lambda_N} \bigl[ a_j^*, a_j^{} \bigr] &- J\sum_{j\in\Lambda_N'} \bigl( a_j^* - a_j^{} \bigr) \bigl( a_{j+1}^* + a_{j+1}^{} \bigr)\\
 & + J   \bigl( a_{N/2-1}^* - a_{N/2-1}^{} \bigr) \bigl( a_{-N/2}^* + a_{-N/2}^{} \bigr) P
\end{split} 
\end{equation}
where $P=\exp( i\pi \sum_{l}a_l^*  a_l^{} )$ is the `parity operator'. It commutes with each term in the Hamiltonian. Therefore $H_N$ is block-diagonal with respect to the direct sum $\caH_N=\caH_N^\mathrm{even} \oplus\caH_N^\mathrm{odd}$ and $P$ acts as (minus) the identity on $\caH_N^\mathrm{even}$ ($\caH_N^\mathrm{odd}$), the space with even (odd) numbers of Jordan--Wigner-fermions. 

\subsection{Fourier-transformation}
We proceed with diagonalizing $H_N$ separately on each $P$-eigenspace by employing a different Fourier-transformation in each of the two cases $p=\mathrm{even/odd}$,
\begin{equation}
\begin{split}
 &\hat{a}_k^*=N^{-\frac{1}{2}} \sum_{j\in\Lambda_N} \e^{ikj}a_j^*\\
 &a_j^*=N^{-\frac{1}{2}} \sum_{k\in K_N^p} \e^{-ikj} \hat{a}_k^*
\end{split} 
\end{equation}
for $j\in\bbZ$, and
\begin{equation}
\begin{split}
 K^\mathrm{even}_N &=\bigl\{\, {\textstyle\frac{2 \pi}{N}}\bigr(n+{\textstyle\frac{1}{2}}\bigr) \,\bigl| \,n={\textstyle -\frac{N}{2},\dots,\frac{N}{2}-1}\bigr\}\\
 K^\mathrm{odd}_N &=\bigl\{\, {\textstyle\frac{2 \pi}{N}}n \, \bigl|\, n={\textstyle -\frac{N}{2},\dots,\frac{N}{2}-1}\bigr\}
\end{split}
\end{equation}	
For $p=\mathrm{even}$ this choice imposes anti-periodic boundary conditions, in particular $a^*_{N/2}=-a^*_{-N/2}$, whereas for $p=\mathrm{odd}$ periodic boundary conditions are more convenient ensuring that $a_{N/2}^*=a_{-N/2}^*$ and $\hat{a}_{-\pi}^{*}=\hat{a}_{\pi}^{*}$.
One then obtains
\begin{equation}
 H_N\bigl\rvert_{\caH_N^\mathrm{even}}= H_N^\mathrm{even}\bigl\rvert_{\caH_N^\mathrm{even}}\quad\text{and}\quad H_N\bigl\rvert_{\caH_N^\mathrm{odd}}= H_N^\mathrm{odd}\bigl\rvert_{\caH_N^\mathrm{odd}}
\end{equation}
for the following operators on the whole Hilbert-space $\caH_N$
\begin{equation}
 \begin{split}
  &H_N^\mathrm{even}=-2J \sum_{0<k\in K_N^\mathrm{even}} \bigl(g+\cos(k)\bigr) \bigl( \hat{a}_{-k}^*\hat{a}_{-k}^{}-\hat{a}_{k}^{}\hat{a}_{k}^{*} \bigl)\\[-3mm]
  &\hphantom{H_N\bigl\rvert_{\caH_N^\mathrm{even}}=-2J \sum_{0<k\in K_N^\mathrm{even}}}+i\sin(k) \bigl( \hat{a}_{k}^{}\hat{a}_{-k}^{}-\hat{a}_{-k}^{*}\hat{a}_{k}^{*} \bigr)\\
  &H_N^\mathrm{odd}=-2J(g+1) \bigl( \hat{a}_0^*\hat{a}_0^{}-{\textstyle \frac{1}{2}}\bigr)-2J(g-1)\bigl(\hat{a}_{-\pi}^*\hat{a}_{-\pi}^{}-{\textstyle \frac{1}{2}}\bigr)\\
  &\hphantom{H_N\bigl\rvert_{\caH_N^\mathrm{odd}}=}-2J \sum_{0<k\in K_N^\mathrm{odd}} \bigl(g+\cos(k)\bigr) \bigl( \hat{a}_{-k}^*\hat{a}_{-k}^{}-\hat{a}_{k}^{}\hat{a}_{k}^{*} \bigl)\\[-3mm]
  &\hphantom{H_N\bigl\rvert_{\caH_N^\mathrm{odd}}=-2J \sum_{0<k\in K_N^\mathrm{even}}}+i\sin(k) \bigl( \hat{a}_{k}^{}\hat{a}_{-k}^{}-\hat{a}_{-k}^{*}\hat{a}_{k}^{*} \bigr)
 \end{split}
\end{equation}

\subsection{Bogoliubov-transformation}
We finally diagonalize $H_N$ by means of a Bogoliubov-transformation on pairs of Jordan--Wigner-fermions with opposite momenta. We define a new set of operators $\{\alpha_k\}$, $k\in K_N^p$, respectively for even or odd parity $p$, and their adjoints through 
\begin{equation}
\begin{split}
 &\alpha_0^{}=\hat{a}_0^*,\quad \alpha_{-\pi}^{}=\hat{a}_{-\pi}^*,\quad \textnormal{for } g>0\\
 &\alpha_0=\hat{a}_{-\pi},\quad \alpha_{-\pi}=\hat{a}_{0},\quad \textnormal{for } g<0
\end{split}
\end{equation}
and 
\begin{equation}
 \begin{pmatrix}
  \hat{a}_k^{}\\\hat{a}_{-k}^*
 \end{pmatrix}
 =U \begin{pmatrix}
  \alpha_k^{}\\\alpha_{-k}^*
 \end{pmatrix},\quad U=\begin{pmatrix}
  \cos(\theta_k/2)&i \sin(\theta_k/2)\\ -i \sin(\theta_{-k}/2)&\cos(\theta_{-k}/2)
 \end{pmatrix}
\end{equation}
for $0<\lvert k\rvert\in K_N^p$, in obvious matrix notation. If the so-called Bogoliubov-angles $\theta_k$ are chosen according to 
\begin{equation}
 \e^{i\theta_k}=-\frac{g+\e^{-ik}}{\lvert g+\e^{-ik} \rvert}
\end{equation}
then $U$ is a unitary matrix and $\{\alpha_k^{},\alpha_k^*\}$, $k\in K_N^p$, satisfy the ({\small CAR}). With respect to these two algebras of creation and annihilation operators $H_N$ simply describes free fermions on each $P$-eigenspace, i.e.
\begin{equation}
 \begin{split}
  &H_N^p=2J \sum_{k\in K_N^p} \bigl\lvert g+\e^{-ik} \bigr\rvert \bigl( \alpha_k^* \alpha_k^{}-{\textstyle \frac{1}{2}} \bigr)
 \end{split}
\end{equation}

\subsection{The ground state as vacuum of free fermions}
The difference of vacuum energies with respect to $H_N^p$ for different values of the parity $p$ vanishes in the limit $N\rightarrow \infty$, since it is (half) the Riemann-sum of the derivative of the periodic function $2 J \lvert g+\e^{-ik}\rvert$. As typical for the new vacuum state after a Bogoliubov-transformation of the above type, it is a superposition of states with possibly several pairs of opposite-momentum fermions with respect to $\{\hat{a}_k^{},\hat{a}_k^{*}\}$. In particular, the (non-degenerate) ground state of $H_N^p$ is an element of $\caH_N^\mathrm{even}$ for both values of the parity $p$. Therefore, the ground state of $H_N^\mathrm{even}$, i.e., the vacuum for the explicitly given $\{\alpha_k^{},\alpha_k^*\}$, $k\in K_N^\mathrm{even}$, is equal to the ground state of $H_N$.

\subsection{Generating function for the transverse magnetization} 
Here we compute the moment generating function $G^n$ for the magnetic moment in the $z$-direction of a chain of $n\geq 1$ sites viewed as (classical) discrete random variable with distribution 
\begin{equation}
 \bbP(x)=\omega \bigl( \funit_x\bigl( {\textstyle \sum_j} \sigma_j^z \bigr) \bigr),\quad x\in\mathrm{sp}\bigl( {\textstyle \sum_j} \sigma_j^z \bigr)
\end{equation}
induced by the (quantum) infinite-volume ground state $\omega$. Note that $\omega$ is translation-invariant as the unique ground state for a translation-invariant interaction and therefore we restrict to evaluating
\begin{equation}
 G_N^{}(\alpha)=\lim_{N\rightarrow \infty} G_N^n (\alpha),\quad G_N^n (\alpha)=\omega_N \left( \exp\bigl({\textstyle \alpha \sum_{j\in \Gamma_n}\sigma_j^z}\bigr) \right), \;\alpha\in\bbC
\end{equation}
for chains of the form $\Gamma_n=\{1,\dots,n\}$. Defining, 
\begin{equation}
 A_j^{}= \left( a_j^*+a_j^{} \right),\quad B_j^\alpha =\left( \e^{-\alpha}a_j^{*}+ \e^{\alpha} a_j^{} \right), \;j\in\Lambda_N
\end{equation}
which are linear combinations also of the transformed creation/ annihilation operators $\alpha_k^* $/$ \alpha_k^{}$, Wick's theorem can be used to evaluate
\begin{equation}
 G_N^n (\alpha)=\omega_N\left( {\textstyle \prod_{j\in\Gamma_n} A_j^{}B_j^\alpha} \right)=\omega_N \bigl( A_1^{}\dots A_n^{}B_n^\alpha \dots B_1^\alpha \bigr)
\end{equation}
in terms of pair-correlations. Using again the explicit form of $\omega_N$ as Fermi-vacuum, one obtains,
\begin{equation}
 \begin{split}
  &\omega_N \left( A_j^{} A_{j'}^{}\right)=\delta_{jj'},\qquad j,j'\in\Lambda_N,\;N\geq 1,\\
  &\omega \left( A_j^{} B_{j'}^\alpha \right)=\lim_{N\rightarrow\infty}\omega_N\left( A_j^{} B_{j'}^\alpha  \right),\qquad j,j'\in\bbZ,\\
  &=\int_{-\pi}^\pi\,\frac{\mathrm{d}k}{2\pi}\,\left( \cosh(\alpha)-\sinh(\alpha) \e^{-i\theta_{ k}} \right)\e^{-ik(j-j')}
 \end{split}
\end{equation}
and therefore the generating function can be written as the determinant of an $n\times n$ Toeplitz-matrix $M_n^\alpha$,
\begin{equation}
 G^n(\alpha)=\det\left( M_n^\alpha \right),\quad \bigl(M_n^\alpha\bigr)_{jj'}=\int_{-\pi}^\pi\,\frac{\mathrm{d}k}{2\pi}\,\phi_\alpha(k)\e^{-ik(j-j')}
\end{equation}
whose entries are the Fourier-coefficients of the function
\begin{equation}
 \phi_\alpha(k)= \cosh(\alpha)-\sinh(\alpha) \e^{-i\theta_{ k}}
\end{equation}
which is called a \textit{symbol} when viewed as function on the complex unit circle $z=\e^{ik}\in\bbC$, $k\in (-\pi,\pi]$. 

\end{appendix}

\bibliographystyle{plain}

\end{document}